\RequirePackage{snapshot}
\newif\ifproofs\proofsfalse\ifproofs\RequirePackage[displaymath,mathlines]{lineno}\fi

\newif\ifsubsections
\subsectionsfalse

\documentclass[]{pcmi}
 
\usepackage[sc]{mathpazo}          
\usepackage{eulervm}               
\usepackage[mathscr]{eucal}		   
\usepackage[scaled=0.86]{berasans} 
\usepackage[scaled=1]{inconsolata} 
\usepackage[T1]{fontenc}

\usepackage[%
	protrusion=true,
	expansion=false,
	auto=false
	]{microtype}

\usepackage{xcolor}
\usepackage{graphicx}
\graphicspath{{./figures/}}

\ifdraft
    \definecolor{linkred}{rgb}{0.7,0.2,0.2}
    \definecolor{linkblue}{rgb}{0,0.2,0.6}
\else
    \definecolor{linkred}{cmyk}{0.0,0.0,0.0,1.0}
    \definecolor{linkblue}{cmyk}{0,0.0,0.0,1.0}
\fi

 \usepackage{amssymb}
 \usepackage[framemethod=TikZ]{mdframed}
 \usepackage{mathtools}
 \mathtoolsset{showonlyrefs}
 \usepackage[english]{}
 \usepackage{comment}
 \usepackage{epigraph}
\PassOptionsToPackage{hyphens}{url} 
\usepackage[
    setpagesize=false,
    pagebackref,
	pdfpagelabels,
	plainpages=false,
    pdfstartview={FitH 1000},
    bookmarksnumbered=false,
    linktoc=all,
    colorlinks=true,
    anchorcolor=black,
    menucolor=black,
    runcolor=black,
    filecolor=black,
    linkcolor=linkblue,
	citecolor=linkblue,
	urlcolor=linkred,
]{hyperref}
\usepackage[backrefs,msc-links,initials,nobysame]{amsrefs}

\customizeamsrefs 

\theoremstyle{plain}
  
  \newtheorem{theorem}[equation]{Theorem} 
  \newtheorem{proposition}[equation]{Proposition} 
  \newtheorem{conjecture}[equation]{Conjecture}

\theoremstyle{definition}
 \newtheorem{definition}[equation]{Definition}
 
  \newtheorem{example}[equation]{Example} 
  \newtheorem{remark}[equation]{Remark}

    \newtheorem{problem}[]{Problem}

         \usepackage{lipsum}
\mdfdefinestyle{MyFrame}{%
    linecolor=black,
    outerlinewidth=1pt,
    roundcorner=5pt,
    innertopmargin=\baselineskip,
    innerbottommargin=\baselineskip,
    innerrightmargin=10pt,
    innerleftmargin=10pt,
    backgroundcolor=gray!6!white}

\tikzset{->-/.style={decoration={
  markings,
  mark=at position #1 with {\arrow{>}}},postaction={decorate}}}
\usetikzlibrary{arrows, positioning, shapes, decorations.pathreplacing, decorations.markings, calc, matrix}

         \usepackage[all]{xy}

 \newcommand{\A}{\mathcal{A}_{reg}}
\newcommand{\M}{\mathcal{M}_{reg}}
\newcommand{\MGC}{\ensuremath{\mathcal{M}_{G_{\mathbb{C}}}}}
\newcommand{\MG}{\ensuremath{\mathcal{M}_{G}}}

 \newcommand{\PP}{\mathbb{P}^1}
 \newcommand{\CE}{\mathcal{E}}
 
 \newcommand{\CO}{\mathcal{O}}
  \newcommand{\C}{\mathbb{C}}
\newcommand{\CP}{\ensuremath{\mathcal P}}
\newcommand{\CX}{\ensuremath{\mathcal X}}
 \newcommand{\R}{\mathbb{R}}
  \newcommand{\SL}{\textrm{SL}}
  \newcommand{\SO}{\textrm{SO}}
   \newcommand{\Sp}{\textrm{Sp}}

\newcommand{\ENa}{\end{equation}}

\begin{document}

%
%
%
%
%
%

\title[Mathematics and physics of Higgs bundles]{ Advanced topics in gauge theory: \\ mathematics and physics of Higgs bundles}

%
%
\author{Laura P. Schaposnik}
\address{Department of Mathematics, University of Illinois at Chicago, IL 60607, USA}
\email{schapos@uic.edu}

%
%
\subjclass[2010]{Primary ????; Secondary ????}
\keywords{Park City Mathematics Institute}

\begin{abstract}
These notes have been prepared as reading material for the mini-course   given by the author  at the {\em 2019 Graduate Summer School} at \href{https://www.ias.edu/pcmi}{Park City Mathematics Institute} -  Institute for Advanced Study. We   begin by introducing Higgs bundles and their main properties (Lecture 1), and then we  discuss the Hitchin fibration and its different uses (Lecture 2).
The second half of the course is dedicated to studying different types of subspaces (branes) of the moduli space of complex Higgs bundles, their appearances in terms of flat connections and representations (Lecture 3), as well as correspondences between them (Lecture 4).\end{abstract}  

%
%
\maketitle
\thispagestyle{empty}

%
%


\tableofcontents

  

\section{Introduction}

Higgs bundles provide a unifying structure for diverse branches of mathematics and physics. The \textit{Dolbeault Moduli space} of $G_\C$-Higgs bundles  $\mathcal{M}_{G_\C}$ has a hyperk\"ahler structure, and through different complex structures it  can be understood as different moduli spaces:
 \begin{itemize}
 \item
Via the non-abelian Hodge correspondence  the moduli space   is   diffeomorphic as a real manifold to the \textit{de Rham moduli space}   of flat connections \cite{cor,6,N1,simpson88,yau}.
\item Via the Riemann-Hilbert correspondence there is an analytic correspondence between the   de Rham moduli space and the \textit{Betti moduli space}   of surface group representations.   

\end{itemize}


\noindent
Through these correspondances, Higgs bundles manifest as both flat connections and representations, fundamental objects which are ubiquitous in contemporary mathematics, and closely related to theoretical physics. Some   examples  are:
\begin{itemize}
\item  Through the Hitchin fibration, $\MGC$ gives examples of hyperk\"ahler manifolds which are \textit{integrable systems} \cite{N2},  leading to remarkable applications in physics, for instance in the works of Gukov, Hitchin, Donagi, Pantev, and Witten (e.g., see \cite{Don93,Don95,dopa,sergei0,witten1,N3,witten}). 
\item  Hausel-Thaddeus \cite{Tamas1}  related Higgs bundles to \textit{mirror symmetry}, and   with Donagi-Pantev \cite{dopa}, presented $\MGC$ as a fundamental  example of mirror symmetry for Calabi-Yau manifolds, whose geometry and topology continues to be  studied. 
\item   Kapustin-Witten \cite{Kap} used  Higgs bundles to obtain a physical derivation of the \textit{geometric Langlands correspondence} through mirror symmetry. Soon after, Ng\^{o} found Higgs bundles key ingredients when proving the fundamental lemma in \cite{ngo}. 
 \end{itemize}
In these Lecture Notes we will focus on Higgs bundles and their moduli spaces, by considering three main aspects of the theory:
\begin{itemize}
\item \textbf{The Hitchin fibration}. We shall review the Hitchin fibration as a tool to understand Higgs bundles, and describe the abelianization introduced by Hitchin in \cite{N2}, and the nonabelianization appearing for certain subspaces of Higgs bundles as in \cite{nonabelian}. 
\item \textbf{Construction of branes}. We shall   construct and study families of branes in the moduli space of Higgs bundles.  In particular,   branes are obtained by imposing different structures on the Riemann surface and the gauge group.   We shall also consider
 geometric structures appearing through   branes, among which are  for example \textit{the  spaces of  hyperpolygons}.
\item \textbf{Correspondences}.   Finally, we will also consider correspondences between Higgs bundles and between branes, in relation to different areas of maths and physics: e.g.,  mirror symmetry,  correspondences arising through group automorphisms, and 3-manifolds.  %
  \end{itemize}

For a longer introduction to spectral data for Higgs bundles the reader may refer to the lecture notes in \cite{NUS}     and \cite{thesis}. Other references shall be mentioned across this Lecture Notes. \\


\noindent \textbf{Bibliography}.  We shall highlight the main references considered, as well as the precise places where the methods used were developed. Since it proves to be very difficult to give a comprehensive and exhaustive account of research in tangential areas, we   restrain ourselves to mentioning related work only when it directly involves methods mentioned in the lectures. The reader should  refer to references within the bibliography for further information (e.g., see references in \cite{ap, ana1, thesis}).  \\


\section{The geometry of the moduli space of Higgs bundles} 

\epigraph{\textit{Proof is the end product of a long interaction between
creative imagination and critical reasoning. Without
proof the program remains incomplete, but without the
imaginative input it never gets started.}}{Sir Michael Atiyah}

%
%

%

We shall dedicate this first  chapter  to review the basic concepts of Higgs bundles and some of their generalizations to principal $G_{\C}$-Higgs bundles,  real Higgs bundles, 
 parabolic Higgs bundles and wild Higgs bundles (for more details refer to \cite{N1,N2,6,cor,simpson88,nit,simpson}). 
\subsection{Higgs bundles for complex groups.} 
Unless specified otherwise, we shall let $\Sigma$ be a compact Riemann surface of genus $g\geq 2$ with canonical bundle $K=T^*\Sigma$.  

  \begin{definition}\label{def:clasical}\label{clas}
   A {\em Higgs bundle} is a pair $(E,\Phi)$ for $E$ a holomorphic vector bundle on $\Sigma$, and  the {\em  Higgs field} $\Phi\in H^{0}(\Sigma,\text{End}(E)\otimes K)$.    \end{definition}

Along these notes, we shall refer to \textit{classical Higgs bundles} as the above Higgs bundles in Definition \ref{clas}. Recall that holomorphic vector bundles $E$ on $\Sigma$  are topologically classified by their rank $rk(E)$  and their  degree $deg(E)$, though which one may define their \textit{slope} as \[\mu(E):= deg(E)/  rk(E).\] 
Then, a vector bundle $E$ is   {\em stable} ({\em or semi-stable}) if for any proper, non-zero sub-bundle $F\subset E$ one has $\mu(F)<\mu(E)$ (or $\mu(F)\leq \mu(E)$).
 It is {\em  polystable} if it is a direct sum of stable bundles whose slope is  $\mu(E)$.
The moduli space of stable bundles of fixed degree $d$ and rank $n$ is an algebraic variety $\mathcal{N}(n,d)$ which can be constructed via  Mumford's Geometric Invariant Theory.
Moreover, for coprime $n$ and $d$, the space is a smooth projective algebraic variety of dimension $n^{2}(g-1)+1$.

\begin{example}
Since line bundles are stable, the space $\mathcal{N}(1,d)$ contains all line bundles of degree $d$, and is isomorphic to the Jacobian variety  $\text{Jac}^{d}(\Sigma)$, an abelian variety of dimension $g$.
\end{example}
Generalising the stability condition of vector bundles to Higgs bundles, we shall consider the following:\begin{definition}
 A vector subbundle $F$ of $E$ for which $\Phi(F)\subset F\otimes K$ is said to be a $\Phi$-{\em invariant subbundle} of $E$.  
 A Higgs bundle $(E,\Phi)$ is said to be
\begin{itemize}
 \item {\em  stable} ({\em semi-stable}) if for each proper $\Phi$-invariant  $F\subset E$ one has \[\mu(F)<~\mu(E)~(equiv. \leq);\]
  \item {\em   polystable} if \[(E,\Phi)=(E_{1},\Phi_{1})\oplus (E_{2},\Phi_{2})\oplus \ldots \oplus (E_{r},\Phi_{r}),\] where $(E_{i},\Phi_{i})$ is stable with $\mu (E_{i})=\mu(E)$ for all $i$.
\end{itemize}
 
 \end{definition}


\begin{remark}\label{invariantsubbundle}
 The characteristic polynomial of $\Phi$ restricted to an invariant subbundle  divides the characteristic polynomial of $\Phi$. Hence, if the characteristic polynomial of the Higgs field is irreducible, one knows that the corresponding Higgs bundle is automatically stable. \label{remstable}
\end{remark}

\begin{example} \label{exa}\label{classic example} Choose a square root $K^{1/2}$ of the canonical bundle $K$, and  a section $\omega$ of $K^{2}$. A family of classical Higgs bundles $(E,\Phi_\omega)$ (as in Definition \ref{clas}) may be obtained by considering the vector bundle  $E=K^{\frac{1}{2}}\oplus K^{-\frac{1}{2}}$ and the Higgs bundle   $\Phi_{\omega}$  given by
\begin{align}
 \Phi_{\omega}=\left(\begin{array}{cc}
0&\omega\\1&0
            \end{array}\right) \in H^{0}(\Sigma, \text{End}(E)\otimes K)
\nonumber
\end{align} 
\end{example}
 
Consider a strictly semi-stable Higgs bundle $(E, \Phi)$. Since $E$ is not stable, it admits a minimal $\Phi$-invariant subbundle $F\subset E$ for which $\mu(F)=\mu(E)$. Then,  the induced pair  $(F,\Phi)$ is stable and the induced quotient the quotient $(E/F, \Phi')$ is semistable. By induction,  one obtains a flag of subbundles
\[F_{0}=0\subset F_{1}\subset \ldots \subset F_{r}=E\]
where $\mu(F_{i}/F_{i-1})=\mu(E)$ for $1\leq i\leq r$, and where the induced Higgs bundles $(F_{i}/F_{i-1}, \Phi_{i})$ are stable. This is the \textit{Jordan-H\"{o}lder filtration} of $E$, and whilst it  is not unique, there is an induced graded object which is unique up to isomorphism
\[\text{Gr}(E,\Phi):=\bigoplus_{i=1}^{r}(F_{i}/F_{i-1},\Phi_{i})\]

Two semi-stable Higgs bundles $(E,\Phi)$ and $(E',\Phi')$ are said to be $S${\em -equivalent } if $\text{Gr}(E,\Phi)\cong \text{Gr}(E',\Phi')$. 
%
%
Following   \cite[Theorem 5.10]{nit} we let $\mathcal{M}(n,d)$ be the moduli space of $S$-equivalence classes of semi-stable Higgs bundles of fixed degree $d$ and fixed rank $n$. The moduli space $\mathcal{M}(n,d)$  is a quasi-projective scheme, and has an open subscheme $\mathcal{M}'(n,d)$ which is the moduli scheme of stable pairs. Thus, every point is represented by either a stable or a polystable Higgs bundle. When $d$ and $n$ are coprime, the moduli space $\mathcal{M}(n,d)$  is a smooth non-compact  variety which has complex dimension $ 2n^{2}(g-1)+2$ and the cotangent space of $\mathcal{N}(n,d)$ over the stable locus is contained in $\mathcal{M}(n,d)$ as a Zariski  open subset.   

One of the most important characterisations of stable Higgs bundles on a compact Riemann surface $\Sigma$ of genus $g\geq2$ is given in the work of Hitchin \cite{N2} and Simpson \cite{simpson88},  and which carries through to more general settings: 

\begin{theorem}
 If a Higgs bundle $(E,\Phi)$ is stable and  $\text{deg} ~E = 0$, then there is a unique unitary connection $A$ on $E$, compatible with the holomorphic structure, such that
\begin{equation}
 F_{A}+ [\Phi,\Phi^{*}]=0~\in \Omega^{1,1}(\Sigma, \text{End}~E), \label{2.1}
\end{equation}
 where $F_{A}$ is the curvature of the connection.  
\end{theorem}
Equation (\ref{2.1}) together with the holomorphicity condition
\begin{equation}
d_{A}''\Phi=0 \label{hit2}
\end{equation}
are called the \textit{Hitchin equations}, where  $d_{A}''\Phi$ is the anti-holomorphic part of the covariant derivative of the Higgs field $\Phi$. Moreover, Hitchin showed that the moduli space of Higgs bundles is a hyperk\"ahler manifold with natural symplectic form $\omega$ defined on the infinitesimal deformations
 $(\dot A,\dot \Phi)$ of a Higgs bundle   
\begin{equation}\label{ch2:2.1}
 \omega((\dot{A}_{1},\dot{\Phi}_{1}),(\dot{A}_{2},\dot{\Phi}_{2}))=\int_{\Sigma}\text{tr}(\dot{A}_{1}\dot{\Phi}_{2}-\dot{A}_{2}\dot{\Phi}_{1}),
\end{equation}
 where $\dot A \in \Omega^{0,1}(\text{End}_{0} E)$ and $\dot\Phi\in \Omega^{1,0}(\text{End}_0 E)$ 
(see \cite{N1, N2} for details). For simplicity, we shall fix $n$ and $d$ and write $\mathcal{M}$ for $\mathcal{M}(n,d)$.

 
The notion of Higgs bundle can be generalized to encompass principal $ G_\C$-bundles, for $ G_\C$ a complex  semi-simple Lie group. For more details, the reader should refer  to \cite{N2}.
  \begin{definition}\label{principalLie}\label{defHiggs}\label{complex}
   A $ G_\C${\em -Higgs bundle}  is a pair $(P,\Phi)$ where  $P$ is a principal $G_\C$-bundle over $\Sigma$, and the Higgs field $\Phi$ is a holomorphic section of the vector bundle  {\em ad}$P\otimes_{\mathbb{C}} K$, for {\em ad}$P$ the vector bundle associated to the adjoint representation. 
  \end{definition}

When $ G_\C\subset GL(n,\mathbb{C})$, a $ G_\C$-Higgs bundle gives rise to a Higgs bundle in the classical sense,  with
some extra structure reflecting the definition of $ G_\C$. In particular, classical Higgs bundles are given by $GL(n,\mathbb{C})$-Higgs bundles.  
\begin{example} The Higgs bundles in Example \ref{exa} have traceless Higgs field, and the determinant  $\Lambda^{2}E$ is trivial. Hence, for each quadratic differential $\omega$ one has an $SL(2,\mathbb{C})$-Higgs bundle $(E,\Phi_\omega)$.
\end{example}
%
For  $ G_\C$ be a complex semisimple Lie group, we denote by $\mathcal{M}_{ G_\C}$  the moduli space of $S$-equivalence classes of   polystable $ G_\C$-Higgs bundles.
Following \cite{ram} one can define stability for principal $ G_\C$-bundles, which when considering groups of type $A,B,C,D$, can be expressed in terms of stability for classical Higgs bundles (see \cite[Section 1.1]{ap} for a comprehensive study). Moreover, the notion of polystability may be carried over to principal $ G_\C$-bundles, allowing one to construct the moduli space of isomorphism classes of polystable principal $ G_\C$-bundles of fixed topological type over the compact Riemann surface $\Sigma$. Since in these notes we shall be working with Higgs pairs which are automatically stable (as in Remark \ref{remstable}),   we shall not dedicate time to recall the main study of stability for principal Higgs bundles. For details about the corresponding constructions, the reader should refer for example to \cite[Section 3]{biswas} and  \cite[Section 1]{ap}.

\begin{example}
An $SO(2n,\mathbb{C})$-Higgs bundle is a pair $(E,\Phi)$, for $E$ a holomorphic vector bundle of rank $2n$ with a non-degenerate symmetric bilinear form $(~,~)$, and    $\Phi\in H^{0}(\Sigma,\textrm{End}_{0}(E)\otimes K)$ the Higgs field  satisfying
\begin{equation}\label{soex}(\Phi v,w)=-(v,\Phi w).\end{equation}
Equivalently, an $SO(2n+1,\mathbb{C})$-Higgs bundle is a pair $(E,\Phi)$ for $E$ a holomorphic vector bundle of rank $2n+1$ with  a non-degenerate symmetric bilinear form $(v,w)$, and  $\Phi$ a Higgs field  in $H^{0}(\Sigma,\textrm{End}_{0}(E)\otimes K)$ which satisfies \eqref{soex}.

When considering the structure of an orthogonal Higgs field $\Phi$, one can see that for a generic matrix $A\in \mathfrak{so}(2n+1,\mathbb{C})$ (and also for $A\in \mathfrak{so}(2n,\mathbb{C})$), its  distinct eigenvalues occur in $\pm \lambda_{i}$ pairs. Thus, the characteristic polynomial of $A$, which shall become important in the next chapters,  must be of the form
\begin{equation}\textrm{det}(x-A)=x(x^{2n}+a_{1}x^{2n-2}+\ldots +a_{n-1}x^{2}+a_{n}).\label{charpolso}\end{equation}
 \end{example}

\subsection{Real Higgs bundles.}

 Higgs bundles for real forms were first studied by N. Hitchin in \cite{N1}, and the results for $SL(2,\mathbb{R})$ were generalised in \cite{N5}, where Hitchin studied the case of $G=SL(n,\mathbb{R})$ appearing through involutions on the Lie algebra, a perspective we shall come back to in later chapters. Using Higgs bundles he counted the number of connected components and, in the case of split real forms, he identified a component homeomorphic to $\mathbb{R}^{\text{dim} G(2g-2)}$ and which naturally contains a copy of a  Teichm\"uller space.

Following \cite{N2}, in order to obtain a $G$-Higgs bundle for a non-compact real form $G$ of $G_\C$ and real structure $\tau$,    the connection $A$ which solves Hitchin equations \eqref{2.1}-\eqref{hit2} needs to satisfy
\begin{equation}
\nabla=\nabla_A+\Phi-\rho(\Phi)
\end{equation}
to have  holonomy in  $G$, where $\rho$ is the compact real structure of $G_\C$.

 Since the connection $A$ has holonomy in the compact real form of $G_\C$, we have $\rho(\nabla_A)=\nabla_A$.   Hence, requiring  $\nabla=\tau(\nabla)$ is equivalent to  $\nabla_{A}=\tau(\nabla_{A})$
and $\Phi-\rho(\Phi)=\tau(\Phi-\rho(\Phi)).$
In terms of the involution $\sigma:=\rho\tau$, these two equalities are given by 
\begin{align}\sigma(\nabla_{A})&=\nabla_{A}\label{invol1}\\
            \sigma(\rho(\Phi)-\Phi)&= \Phi-\rho(\Phi).\label{invol2}
\end{align}

From the above, $\nabla$ has holonomy in the real form $G$ if $\nabla_{A}$ is invariant under  $\sigma$, and $\Phi$ anti-invariant. 
 Hence for $G$ a real form of a complex semisimple Lie group $G_\C$, we may construct $G$-Higgs bundles as follows. For $H$ the maximal compact subgroup of $G$, we have seen that the Cartan decomposition of $\mathfrak{g}$ is given by
$\mathfrak{g}=\mathfrak{h}\oplus \mathfrak{m},$
for $\mathfrak{h}$ the Lie algebra of $H$, and $\mathfrak{m}$ its orthogonal complement. This induces a decomposition for the  complex Lie algebras as  
$\mathfrak{g}_\C=\mathfrak{h}_\C \oplus \mathfrak{m}_\C.$
Through the induced isotropy representation 
$\textrm{Ad}|_{H^{\mathbb{C}}}: H^{\mathbb{C}}\rightarrow GL(\mathfrak{m}^{\mathbb{C}}), $
one has  a concrete description of $G$-Higgs bundles (for more details, see for example \cite{brad1}):

\begin{definition}\label{real}
 A \textrm{principal} $G$\textrm{-Higgs bundle } is a pair $(P,\Phi)$ where
 $P$ is a holomorphic principal $H^{\mathbb{C}}$-bundle on $\Sigma$, and 
  $\Phi$ is a holomorphic section of the bundle $P\times_{Ad}\mathfrak{m}^{\mathbb{C}}\otimes K$.
 \end{definition}

\begin{example}\label{expq}
Consider the non-compact real form $\textrm{SU}(p,q)$ of $\SL(p+q,\C)$ with Lie algebra 
\[\mathfrak{su}(p,q)=\left\{
\left(
\begin{array}{cc}
 Z_{1}&Z_{2}\\
\overline{Z}^{t}_{2}&Z_{3}
\end{array}
\right)~\left|
\begin{array}{c}
Z_{1}, Z_{3}~ {~\rm ~ skew~Hermitian~of ~order  ~}p \textrm{~and ~} q, \\
\textrm{Tr}Z_{1}+\textrm{Tr}Z_{3}=0~,~Z_{2}{~\rm arbitrary}
\end{array}\right.
\right\},\] whose complexification has Cartan decomposition  %
\[\mathfrak{su}(p,q)^{\mathbb{C}}=\mathfrak{sl}(p+q,\mathbb{C})=\mathfrak{h}^{\mathbb{C}}\oplus \mathfrak{m}^{\mathbb{C}},\]
where $\mathfrak{m}^{\mathbb{C}}$ corresponds to the off diagonal elements of $\mathfrak{sl}(p+q,\mathbb{C})$. 
The centre of $SU(p,q)$ is $U(1)$, and its  maximal compact subgroup is    $H=S(U(p)\times U(q)),$ whose complexified  Lie group is 
\[H^{\mathbb{C}}=\{(X,Y)\in GL(p,\mathbb{C})\times GL(q,\mathbb{C}): \textrm{det} Y = (\textrm{det}X)^{-1}\}.\] Hence,   an $SU(p,q)$-Higgs bundle over $\Sigma$ is a pair $(E,\Phi)$ where
 $E=V_{p}\oplus V_{q}$ for $V_{p},V_{q}$ vector bundles over $\Sigma$ of rank $p$ and $q$ such that $\Lambda^{p}V_{p}\cong \Lambda^{q}V_{q}^{*}$, and the Higgs field $\Phi$ is given by
        \begin{equation}
	        \Phi={ \left( \begin{array}
	          {cc} 0&\beta\\
	\gamma&0
	         \end{array}\right), \label{hig}}
        \end{equation}
for $\beta:V_{q}\rightarrow V_{p}\otimes K$ and $\gamma:V_{p}\rightarrow V_{q} \otimes K$.  
 
\end{example}

%
%

%

For more details about the construction of real Higgs bundles and the study of their geometric properties, the reader might be interested to see the lecture notes and thesis \cite{Go2,thesis,NUS,peon,ap}, as well as the following -- certainly non-exhaustive- list of papers which we have found useful in the past (and references therein): \cite{N5,brad,xia2,Richard,xia,umm,gothen1,andre1,nonabelian,cayley,brad3,ortlau,brad19,andre2,brad1,brian2,brian3,bra1}
In particular, the reader should refer to \cite{GP09} for the Hitchin-Kobayashi type correspondence for real forms.
 
\subsection{Parabolic   Higgs bundles.}\label{secpara}

The notion of Higgs  bundles can be generalized in several directions, and one which has caught the attention of mathematicians and physicists during the last decades is the case of \textit{parabolic Higgs bundle}, following the generalization of vector bundles to parabolic bundles \cite{meta}. 
 On a compact Riemann surface $\Sigma$ of genus $g$ with marked points $p_1, \ldots, p_n$ satisfying  $2g-2+n\geq 0$, we denote by $D$ the effective divisor $D=p_1+\ldots+p_n$.  Then, from \cite{meta} a \textit{parabolic vector bundle}  $\boldsymbol{E}$  on $\Sigma$ is an algebraic rank $r$ vector bundle $E$ on $\Sigma$ together with a \textit{parabolic structure}, which is a (not necessarily full) flag for the fibre of $E$ over the marked points
\begin{equation}0=E_{p_i, 0}\subset E_{p_i, 1}\subset E_{p_i,2}\subset\ldots\subset E_{p_i, r_i}=E|_{p_{i}},\end{equation}
together with a set $\alpha_i=(\alpha_{i,0},\ldots, \alpha_{i,r_i})$ of \textit{parabolic weights} associated to each marked point $p_i$ given by real numbers satisfying
\begin{equation}1= \alpha_{i,0}>\alpha_{i,1}>\ldots>\alpha_{i,r_i}\geq 0.\label{parwei}\end{equation}

A \emph{parabolic} holomorphic map is a map $f:E\to E'$ between parabolic bundles  for which 
$\alpha_{i}(p_j)>\alpha'_{k}(p_j)$ implies
$f(E_{p_j,i})\subset E'_{p_j,k+1}$  for all $p_j\in D$, where  
 $\alpha'_{k}(p_j)$ are the weights on $E'$.  
The parabolic degree   of a parabolic bundle $E$ are then defined as
\begin{equation} \textrm{par}\deg(E)= \deg(E)+\sum_{j=1}^n\sum_{i=1}^{r_j}m_{j,i}\alpha_{j,i},\end{equation}
where  $m_{j,i} =\dim (E_{p_j,i}/E_{p_j,i+1})$, and the parabolic slope of the parabolic bundle is defined as $\textrm{par}\mu(E)=  \textrm{par}\deg(E)/\textrm{rk}(E)$. 
A parabolic bundle is called (semi)-stable if for every parabolic
subbundle $F$ of $E$, the parabolic slope satisfies $ \textrm{par}\mu(F)\le
 \textrm{par}\mu(E)$ (resp. $ \textrm{par}\mu(F)< \textrm{par}\mu(E)$). Through the parabolic  stability for $\boldsymbol{E}$ one can construct the moduli space of semi-stable parabolic vector bundles \cite{meta}, usually denoted by $\mathcal{N}_\alpha$, which is a normal projective
variety, non-singular when the weights are generic.

One can generalize the notion of parabolic vector bundles and consider \textit{parabolic Higgs bundles}, or PHG, as was first done by C.~Simpson in \cite{carlos}.  
\begin{definition}A \textit{(classical) parabolic Higgs bundle} is a pair $(\boldsymbol{E}, \Phi)$ where\begin{small}
\begin{itemize}
\item $\boldsymbol{E}$ is a parabolic vector bundle as defined above, 
\item the Higgs field $\Phi\in H^{0}(\Sigma, \textrm{End}(E)\otimes K_{\Sigma}(D))$ satisfies
$ \Phi|_{p_{i}} (E_{i,j})\subset E_{i, j}\otimes  K_{\Sigma}(D)|_{p_{i}}.$\end{itemize}
 \end{small}
 \end{definition}

The notions of stability extend to parabolic Higgs bundles, and through these one obtains  the moduli space $\mathcal{P}_{\alpha}$ of $\alpha$-semi-stable parabolic Higgs bundles of degree $d$ and rank $r$, which was first constructed by Yokogawa in \cite{Yoko} (see also \cite{Yoko2}). This space hasdimension
$  (2g-2)r^2 +2 + nr(r-1),$  is a normal quasi-projective variety, and was
constructed as a hyperk\"ahler quotient using gauge theory  in \cite{konno}, and topological properties of this moduli
space for rank  $r=2$ parabolic Higgs bundles were studied through the Morse theoretic approach of Hitchin in \cite{Yoko}.
\begin{example}\label{expqpara}
A \emph{$U(p,q)$-parabolic Higgs bundle} on $\Sigma$ is a parabolic
Higgs bundle $(E,\Phi)$ such that $E=V\oplus W$, where $V$ and $W$
are parabolic vector bundles of rank $p$ and $q$ respectively, and
  \[
  \Phi=\left(\begin{array}{ll}  0 & \beta \\ \gamma & 0
  \end{array}\right) :(V\oplus W)\to (V\oplus W)\otimes K(D),
  \]
where the non-zero components $\beta:W\to V\otimes K(D)$ and
$\gamma: V\to W\otimes K(D)$ are  parabolic morphisms and $D$ is an effective divisor on $\Sigma$ (the reader should compare this to Example \ref{expq}).
\end{example}
 
Parabolic Hitchin systems can be considered to allow simple poles in the gauge and Higgs fields at marked points on a Riemann surface, whilst higher order poles lead us \textit{Wild Higgs bundles}, which we shall consider next.

 \subsection{Wild Higgs bundles.}
Across these notes we shall   limit ourselves to the study of Wild Higgs bundles  $(E,\Phi)$ over the projective line $\PP$. For a fixed effective divisor $D=p_1+\ldots+p_n$ in $\mathbb{P}^1$, we take these pairs  as  composed of  a vector bundle $E$ and a  $\CO_{\PP}(D)$ valued endomorphism  $\Phi:\CE\rightarrow \CE \otimes \CO_{\PP}(D)$ called the Higgs field. 
 The divisor $D$ controls where $\Phi$ is allowed to have poles, so it is sometimes referred to as the \textit{polar divisor}, and in this setting one   considers  stable, integrable
connections with irregular singularities of the form 
\begin{equation}
\mathcal{A}=d+T_n \frac{dz}{z^n}+ \ldots +T_1 \frac{dz}{z}, 
\end{equation}
for $n>1$, and stable parabolic Higgs bundles $(E,\Phi)$ as in the previous section,  where the 
Higgs field has polar parts, e.g., 
\begin{equation}\label{tes}
T_n \frac{dz}{z^n}+ \ldots + T_1 \frac{dz}{z},
\end{equation}
and is defined such that the fibre of $E$ at $p_j$ is preserved by the residue $\textrm{res}(\Phi,p_j)$, and where  $T_1, T_2, \ldots , T_n$ are elements of the Lie algebra of the group\footnote{One usually assumes that the matrix $T_n$ is diagonalizable with distinct eigenvalues; this rules out nilpotence which rules out the possibility of a higher order regular singularity and also allows one to define the formal diagonalizing gauge transformation taking a connection to its irregular type given by (1.15).}.  When considering Higgs bundles of low rank, we shall sometimes take the notation of   \cite{BoalchIso} and write a complexified connection as:
\begin{equation}
\mathcal{A}= d Q +T_1 \frac{dz}{z},\label{par1}
\end{equation}
for $Q$ a square diagonal matrix of meromorphic connections\footnote{One typically applies a formal diagonalizing procedure before defining $Q$.}, where varying $Q$ in \ref{par1}, is equivalent to varying $(T_n,T_{n-1}, \ldots, T_2)$ in \eqref{tes}.

\begin{definition}\label{mono11}
The \textbf{formal monodromy} of $\mathcal{A}$ is
$M_0:=e^{2\pi i T_1}.$
Moreover, we say that $T_1$ is the \textbf{exponent} of the formal monodromy.
\end{definition}

Following  SimpsonÕs construction in \cite{simpson} for parabolic Higgs bundles, it is natural to assume  
that the connections and Higgs fields are holomorphically gauge equivalent to ones with diagonal polar
parts.   Similar to the regular case, the moduli space of these wild Higgs bundles  can be obtained as a hyperk\"ahler quotient  \cite{Boa12,Boa14}. 
In these notes we shall follow the notation of  \cite{BoalchAnnal} where a wild Higgs bundle of type $(m, r_1 , r_2 , \ldots)$ is a Higgs bundle with $m$ poles of orders $k_i:=r_i+1$. To understand the moduli spaces of Higgs bundles, in what follows we shall restrict our attention to rank 2 Higgs bundles with poles over $\mathbb{P}^1$.

\begin{remark}It is interesting to consider   Higgs bundles whose number of poles with multiplicities is exactly 4 since those are closely related to solutions to Painlev\'e  equations:
 (A)  type $(1,3)$ gives Painlev\'e  II -- see Problem I.(5).iii.; 
 (B)  type $(2,1,1)$ gives  Painlev\'e  III -- see Example \ref{ex22};
 (C) type $(4,0,0,0,0)$ gives  Painlev\'e  VI -- see Section \ref{secpara}). \label{Prem}
 \end{remark}

 For rank two Higgs bundles, from  \cite[Remark 9.12]{BoalchAnnal} one obtains moduli spaces of complex dimension two when
 $(m, r_1 , r_2 , \ldots)=(4, 0, 0, 0, 0)$ as in (C) of Remark \ref{Prem}, $(3, 1, 0, 0), (2, 1, 1)$ as in (B) of Remark \ref{Prem}, $(2,
2, 0),$ and  $(1, 3)$ as  in (A) of Remark \ref{Prem}. In order to illustrate the main ingredients in the study of wild Higgs bundles, in what remains of the chapter we shall consider some particular examples. The reader should refer to Boalch's papers mentioned above for more details and complete proofs, as well as \cite[Appendix A]{tbrane} for further examples appearing in relation to String theory.  Finally, for open problems in the area, a non-exhaustive list of questions relating Wild Higgs bundles to singular geometry is given in \cite{SIGMA}, and for more on the relation with Painlev\'e's equations see for example \cite{SZ} and references therein.

\subsubsection*{Higgs bundles with one pole}
In what follows we shall give a construction of the main ingredients introduced in \cite{BoalchIso} by first studying rank 2 Higgs bundles on $\mathbb{P}^1$ with one double pole on a fixed point $p_1=0$ ( see also \cite[Appendix A]{Anderson:2017rpr}). In this setting, the connection in \eqref{par1} is
\begin{equation}
\mathcal{A}=  dz \left(\frac{T_2}{z^2}+\frac{T_1}{z}\right), \label{L_wit1}
\end{equation}
where $T_2=\frac{dQ}{dz} z^2$ and $T_1$  is a constant diagonal matrix.  In particular, we write \begin{equation}
Q=\left(\begin{array}
{cc}
q_1(z)&0\\0&q_2(z)
\end{array}\right), \label{par2}
\end{equation}
and define $q_{ij}(z)$ as the leading term of $q_i(z)-q_j(z)$. For $q_1(z) -q_2(z)=a/z+b$, one has that 
 \begin{equation}
 q_{12}(z)=a/z,~\textrm{and }~
 q_{21}(z)=-a/z.\end{equation}

The set of \textit{anti-Stokes directions} $\mathbb{A} \subset S^1$ is composed of the directions $\textrm{d} \in S^1$ for which either $q_{12}(z) \in \mathbb{R}_{<0}, ~{~\rm or~}q_{21}(z) \in \mathbb{R}_{<0}
$ 
for $z$ on the ray specified by $\textrm{d}$.
These are the directions along which $e^{(a/z + b)}$ (respectively $e^{-(a/z + b)}$) decays most rapidly as $z$ approaches 0, 
and since
 if
$\pm ~\frac{a}{z} \in \mathbb{R}_{<0}$ then $\mp~\frac{a}{z} e^{-i\pi}\in \mathbb{R}_{<0}$, and thus $\# \mathbb{A}:=r= 2.$
Note that if $q_{12}(z)\in  \mathbb{R}_{<0}$ for a direction $\textrm{d}$, then $q_{21}(z)\not \in  \mathbb{R}_{<0}$ for that direction (and vice versa). We shall refer to a direction as a \textit{half-period}:  for $l:=r/2$,   a \textit{half-period}  is  an $l$-tuple $\textrm{d}=(\textrm{d}_1,\ldots, \textrm{d}_l)\subset \mathbb{A}$ of consecutive anti-Stokes directions. The \textit{multiplicity} $\textrm{Mult}(d)$ of $\textrm{d}$ is the number of roots supporting $\textrm{d}$.  When weighted by their multiplicities, the number
of anti-Stokes directions in any half-period is $1=\textrm{Mult}(\textrm{d}_1)+\ldots+\textrm{Mult}(\textrm{d}_l).$  
In our setting of \eqref{L_wit1},  letting $a=\tilde r e^{i \tilde \theta}$ there are only two half periods: 
\begin{itemize}
\item From $q_{12}$  a direction $\textrm{d}_1:= e^{i( \tilde \theta -\pi)},$ for which taking $z=  r' e^{i(  \tilde \theta-\pi)}$ one has $q_{12}(z)=-\frac{\tilde r}{ r'}\in \mathbb{R}_{<0};$
\item From $q_{21}$   a direction $\textrm{d}_2:=e^{-i\pi} \textrm{d}_1=e^{i \tilde \theta},$ for which taking $z=  r' e^{i\tilde \theta}$ gives $q_{21}(z)= -\frac{\tilde r }{r'}\in \mathbb{R}_{<0}.$
\end{itemize}
 
For $d$ an anti-Stokes direction, the \textit{roots} of $\textrm{d}$ are
\[\textrm{Roots}(\textrm{d}) :=\{(ij)~ | ~q_{ij}(z) \in  \mathbb{R}_{<0}~{\rm~ along~} \textrm{d}\}.\] Note that  the Stokes directions $ \textrm{d}_{1}$ and $ \textrm{d}_{2}$ define the orderings
$q_1<_{ \textrm{d}_{1}} q_2,$ and
$q_2<_{ \textrm{d}_{2}}q_1.
$
Moreover, the directions define the \textit{Stokes sectors} ${\textrm{Sect}_i} :=(\textrm{d}_i , \textrm{d}_{i+1} )$ for $i$ mod $\ell$.  %
 For the  Higgs bundle arising from \eqref{L_wit1} the roots are $
\textrm{Roots}(\textrm{d}_1) =\{(12)\}$ and $
\textrm{Roots}(\textrm{d}_2) =\{(21)\}$.

\begin{remark}\label{rem4}For a Higgs bundle arising from \eqref{par1} with one pole of order 3 at zero, writing
$q_1 -q_2=\frac{a}{z^4}+\frac{b}{z^3}+\frac{c}{z^2}+\frac{d}{z}+e,$ leads to   $q_{12}(z)=\frac{a}{z^3}$ and $q_{21}(z)=-\frac{a}{z^3}.$ This system has 6 anti-Stokes directions, and we shall come back to it in the problem set. \end{remark}
  
%

 The \textit{group of Stokes factors} associated to a direction $d$  is the group
\[\mathbb{S}to_d(\mathcal{A}):=\{K\in  G ~| ~(K)_{ij}=\delta_{ij} ~{~\rm unless~} ~(ij) ~{\rm~is~ a ~root ~of~} d\},\]
  which for rank 2 Higgs bundles is a unipotent subgroup of $G=GL(2,\mathbb{C})$ of dimension 1.
For all directions $ \textrm{d}_{1}$ and $ \textrm{d}_{2}$ that we have, the group of \textit{Stokes matrices}\footnote{
It should be noted that Boalch and other authors (e.g. see \cite{charles}) call \textit{Stokes matrices }  a product of a half-period's worth of \textit{Stokes factors}, the objects in \eqref{factor} --  in the rank 2 case these two definitions coincide. 
} are 1-dimensional subgroup  
\begin{equation}\mathbb{S}to_{ \textrm{d}_{1}}(\mathcal{A})=\left\{  \left(\begin{array}{cc}
1& \tilde z\\
0&1
\end{array}\right)~\textrm{for ~} \tilde z \in \mathbb{C}\right\}~ {~\rm and~}~\mathbb{S}to_{ \textrm{d}_{2}}(\mathcal{A})=\left\{  \left(\begin{array}{cc}
1& 0\\
\tilde z&1
\end{array}\right)~\textrm{for ~} \tilde z \in \mathbb{C}\right\}.\label{factor}\end{equation}
  
  To study the moduli space of Wild Higgs bundles one needs to consider fundamental solutions in each Stokes sector, which are then extended using the Stokes matrices, which can then be thought of as  the transition matrices between the canonical fundamental solutions. However,  we shan't deepen into this, and rather refer the reader to Boalch's papers mentioned before (e.g. see \cite{Boa12,BB04,Boa14,BoalchAnnal,BoalchIso}).

\subsubsection{Higgs bundles with multiple poles.}\label{multiple}

When considering   Higgs bundles with different poles over a divisor $D$, more data needs to be taken into account in order to build their moduli space. In what follows we shall mention a few ingredients to give the reader a flavour of the theory by studying rank 2 Higgs bundles with four poles. We shall consider here a divisor $D=\{p_1,\ldots,p_4\}$,  take $V$ a homomorphically trivial bundle over $\mathbb{P}^1$, and  let $\nabla=d-\mathcal{A}$ be a meromorphic connections with poles on $D$, for 
\begin{equation}
\mathcal{A}=\sum_{i=1}^{4}\left( T^i_{k_i} \frac{dz}{(z-p_i)^{k_i}}+\cdots +T^i_{1} \frac{dz}{(z-p_i)^{1}}\right),
\end{equation}
where $T^i_j$ are diagonal $2\times 2$ matrices. 
Here the connection $\nabla$ has a pole of order $k_i$ at $p_i$, and assuming that $T^1_1+\cdots+T^4_1=0$, it has no other poles.   
To deal with this type of Higgs bundles  we choose for each $i$ disjoint open discs $\textrm{D}_i$ on $\mathbb{P}^1$ with center $p_i$, and  coordinate $z_i$ on $\textrm{D}_i$ vanishing at $p_i$. Then, the local theory described before for one single pole  is repeated on each such disc:  Stokes matrices corresponding to distinct poles are related to each other by \textit{connection matrices}, which are transition matrices  defined in terms of local fundamental solutions.  These connection matrices are associated to what is known as a \textit{tentacle}.

To complete the description of the monodromy data of an irregular connection with multiple poles, one needs the Stokes matrices and formal monodromy at each pole along with the connection matrices relating fundamental solutions in neighborhoods of distinct poles,  making choices of points and paths between them. These choices are encoded in a choice of the so-called \textit{\textbf{tentacle}}:
\begin{enumerate}
\item A point $a_j$ in some sector at $p_j$ between two anti-Stokes rays\footnote{Note that a choice of tentacle implies a labelling $P_j$ of the permutation associated to $p_j$ (which for simple poles is the identity matrix). Hence, one must either fix various conventions about how those choices are to be made, or sacrifice some of the geometry by allowing for cyclic permutations of the monodromy data. },
 for each  $j=1, \ldots, 4$.
\item A base-point $p_0 \in \mathbb{P}^1/\{p_1,\ldots, p_4\}$ and a path \[\gamma_j: [0, 1] \rightarrow \mathbb{P}^1/\{p_1,\ldots, p_4\},\] from $p_0$ to $a_j$ for each  $j$, such that the loop
\[
(\gamma^{-1}_4\beta_4\gamma_4)(\gamma^{-1}_3\beta_3\gamma_3)(\gamma^{-1}_2\beta_2\gamma_2)(\gamma^{-1}_1\beta_1\gamma_1)
\] based at $p_0$ is contractible in $\mathbb{P}^1/\{p_1,\ldots, p_4\}$, for $\beta_j$   a loop in $\textrm{D}_j/\{p_j\}$ based at $a_j$ around $p_j$ once in a positive sense.
\end{enumerate}

 As an example, consider   a rank 2 Higgs bundles with 2 poles of order 2 and its tentacle in Figure~\ref{wild10}, where the grey labels and dotted paths   have been chosen to define the theory:
  \begin{figure}[H]
 \centering
  \includegraphics[width=0.8\textwidth]{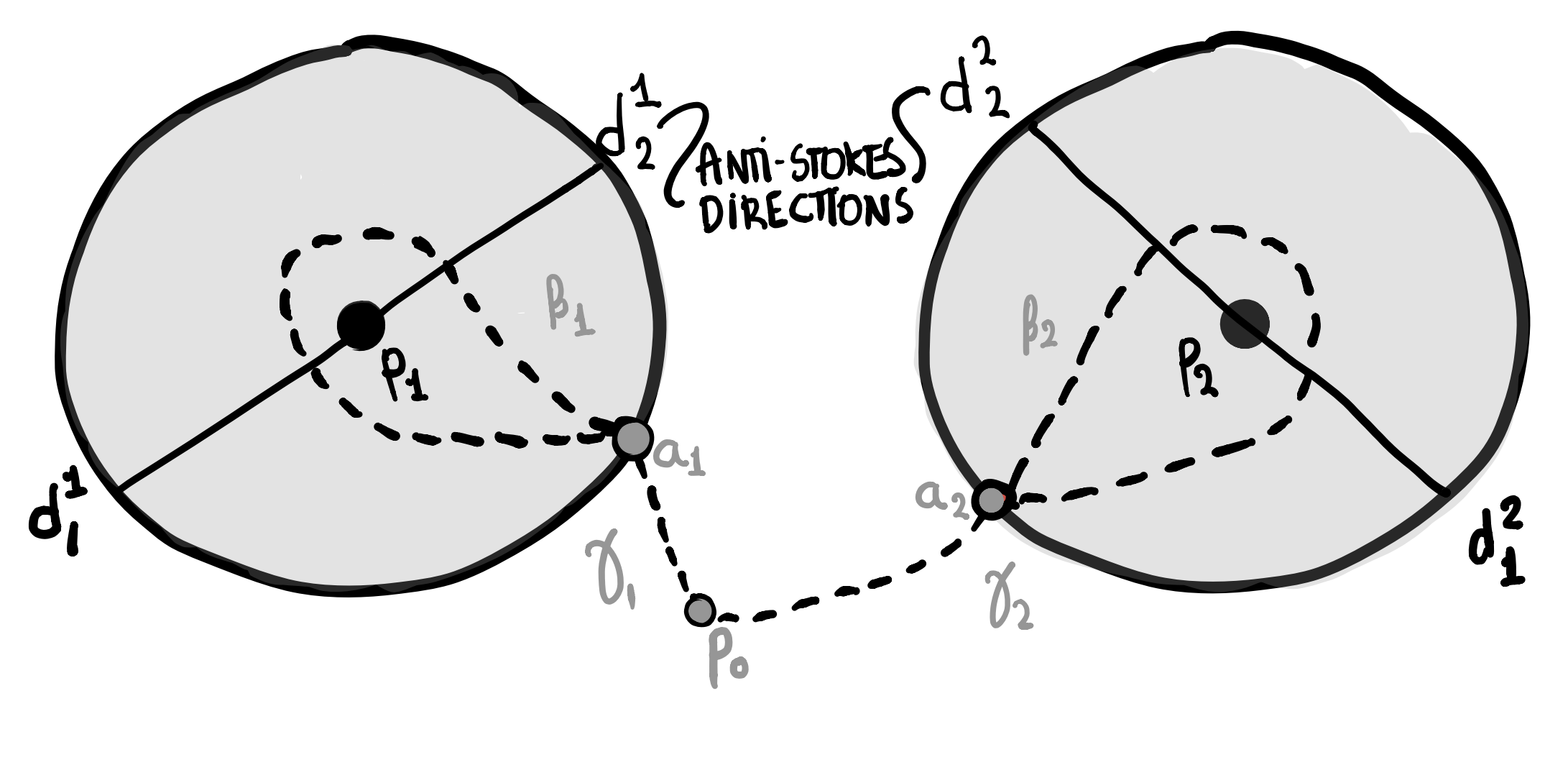} 
\caption{Tentacles for two poles of order 2.}\label{wild10}
\end{figure}

 \label{ex22}
  
We shall conclude by noting that  considering the collision of points, one has the following conjecture about the relation of Wild Higgs bundles and parabolic Higgs bundles:
\begin{conjecture}[\cite{Anderson:2017rpr}]
There exists a flat morphism between the parabolic moduli space  and the wild moduli space in the case of singular parabolic Higgs bundle with $n$-simple poles in their connection, and that of  wild Higgs bundles with a single, higher order pole of order $n$.\end{conjecture}

%

%
%
%
%
\subsection{Problem set I.}

This problem set is meant to make the  reader be acquainted with the different ways of encoding the data defining Higgs bundles in their various forms. They provide a mixture of problems from the lectures as well as some additional ones.

\begin{problem} \textbf{Complex Higgs bundles}. As mentioned before, through Definition \ref{complex} one can define principal $G_\C$-Higgs bundles which may be described in terms of a pair $(E,\Phi)$ of a vector bundle $E$ and a holomorphic Higgs filed $\Phi:E\rightarrow E\otimes K$ with some extra conditions reflecting the nature of the group:\begin{itemize}
\item[i.] Give a description of $SO(2p+1,\C)$-Higgs bundles in terms of pairs $(E,\Phi)$ as in Definition \ref{clas} and list the extra conditions that need to be satisfied. 
\item[ii.] Show that if $\lambda$ is an eigenvalue of an $Sp(2p,\C)$-Higgs bundle $(E,\Phi)$, then necessarily so is $-\lambda$. 

%
%
%
 
\end{itemize}
 
\end{problem}

\begin{problem} \textbf{Real Higgs bundles}. In order to understand $Sp(2p,2q)$-Higgs bundles, we shall begin by considering the structure of the corresponding Lie algebra:

\begin{enumerate}  \item[i.] Give a $4\times4$ matrix expression (with matrix entries) of the non-compact real form $\mathfrak{sp}(2p,2q)$ of the complex symplectic Lie algebra.

 \item[ii.] Show that $\mathfrak{m}^{\mathbb{C}}$ for $\mathfrak{sp}(2p,2q)=\mathfrak{h}\oplus  \mathfrak{m}$  as in Example \ref{expq} can be expressed as subset of certain off-diagonal matrices.

\item[iii.]Prove that if a Higgs bundle $(E,\Phi)$ is stable, then for $\lambda\in \mathbb{C}^{*}$ and $\alpha$ a holomorphic automorphism of $E$, the induced Higgs bundles $(E,\lambda \Phi)$ and $(E,\alpha^{*}\Phi)$ are stable.


\end{enumerate}


\begin{itemize}
\item[iv.] A definition of $Sp(2p,2q)$-Higgs bundles as pairs $(E,\Phi)$ satisfying extra conditions which reflect the nature of the group. 

%


\item[v.] A description of the characteristic polynomial of a generic $Sp(2p,2q)$-Higgs field. 
\end{itemize}

 \end{problem}
\begin{problem}\textbf{Real Higgs bundles}. 
Let   $\mathfrak{g}_{c}$ be a complex Lie algebra with complex structure $i$, whose Lie group is $ G_\C$.  Recall that a  \textit{real form} of $\mathfrak{g}_{\C}$ is a real Lie algebra which satisfies
$\mathfrak{g}_{\C}=\mathfrak{g}\oplus i\mathfrak{g}.$
Given a real form $\mathfrak{g}$ of $\mathfrak{g}_{\C}$, an element  $Z\in \mathfrak{g}_{\C}$ may be written as $Z=X+iY$ for $X,Y\in \mathfrak{g}$. The mapping
$ X+iY\mapsto X-iY$
is called the \textit{conjugation}  with respect to $\mathfrak{g}$.  Moreover, any real form $\mathfrak{g}$ of $\mathfrak{g}_{\C}$ is given by the fixed points set of an antilinear involution $\tau$ on $\mathfrak{g}_{\C}$.  i.e., a map satisfying
\begin{align}
\tau(\tau (X))&=X, &\tau(zX)&=\overline{z}\tau(X),\\ \tau(X+Y)&=\tau(X)+\tau(Y),&\tau([X,Y])&=[\tau (X),\tau (Y)]
\end{align}
 for $X,Y\in \mathfrak{g}_{\C}$ and  $z\in \mathbb{C}$. 
  In particular the conjugation with respect to $\mathfrak{g}$ satisfies these properties. From the above definitions, obtain the following:
\begin{itemize}

\item[i.] Obtain an involution $\sigma$ on $\mathfrak{g}_\C$ such that:\begin{itemize}
\item the fixed point set $\mathfrak{g}_\C^{\sigma}$ of $\sigma$  is given by the complexification of the maximal compact subalgebra $\mathfrak{h}$ of $\mathfrak{g}$;
 \item the anti-invariant set under the involution $\sigma$ is given by $\mathfrak{m}^{\mathbb{C}}$, where $\mathfrak{g}=\mathfrak{h}\oplus \mathfrak{m}$ is the Cartan decomposition.
\end{itemize}

\item[ii.]  A description of the involution on $\mathfrak{gl}(n,\C)$ and  $\mathfrak{sl}(n,\C)$ fixing the unitary groups with signature $\mathfrak{u}(p,q)$ and $\mathfrak{su}(p,q)$. 
\item[iii.] A description of $Sp(2p,2q)$-Higgs bundles as the fixed point set of some involution on the whole moduli space of $Sp(2p+2q,\C)$-Higgs bundles, induced by actions on the Lie algebra, using  the Hitchin equations and question  \textbf{Problem 2} above. 
\end{itemize}

\end{problem}
\begin{problem} \textbf{Parabolic Higgs bundles}. When considering $GL(2,\C)$-Higgs bundles, recall that Hitchin  established in \cite{N2} a correspondence between equivalence classes
of irreducible $GL(2, \C)$ representations of $\pi_1(\Sigma)$ and isomorphism classes of $GL(2,\C)$-Higgs bundles of degree zero. A similar correspondence was later obtained by Simpson in \cite{carlos} between parabolic Higgs bundles and and the fundamental group of $\Sigma/\{p_1,\ldots,p_n\}$. Though those identifications, particular properties of parabolic Higgs bundles and their corresponding representations can be obtained, as was done for example in \cite{indranil}. 
\begin{itemize}
\item[i.] Give an example of a parabolic stable Higgs bundle $(E, \Phi)$ of rank 2 which has  parabolic degree zero.
\item[ii.] Give an example of a parabolic stable $U(3,1)$-Higgs bundle $(E, \Phi)$, using the descriptions of Example \ref{expqpara}.

\item[iii.] Give an expression for the flat connection corresponding to the Higgs bundle $(E, \Phi)$ in \textbf{Problem 4.i.}, and modify the pair until it can be shown that   the holonomy of the flat connection is contained (after conjugation) in $SL(2, \R)$.
\item[iv.] From \cite{steve}, recall that a \emph{minimal parabolic structure} on $E$ is a choice at each $p_i$ of a partial flag $0 \subset L_i \subset E_{p_i}$, where $L_i$ is a line. A \emph{minimal parabolic Higgs field} is an $\CO_X$-linear map $\phi \in H^0(X, \textrm{End}(E) \otimes K_X(D) )$, such that at each point $p_i$, the residue of $\phi_i$ of $\phi$ is strictly triangular with respect to the flag: $\phi(E_{p_i}) \subseteq L_i, \phi_i(L) = 0.$ Construct an example of a rank 3 \emph{minimal parabolic Higgs field}. 
\item[v.] Show that   if minimal parabolic Higgs field $\phi$ as in \textbf{Problem 4.iv.} is strictly parabolic, then the residues $\phi_i$ are nilpotent of order $2$,  and they lie in the closure of the minimal nilpotent orbit of $\mathfrak{sl}_r$.


\end{itemize}
\end{problem}

\begin{problem}
\textbf{Wild Higgs bundles}. Consider a rank 2 Higgs bundle with 1 pole of order 3 at $p_2$ and one of order 1 at $p_1$. For this, take diagonal matrices $^jA_i$ and \begin{equation}
\nabla:=d-   \left( ^1A_{1} \frac{dz}{(z-p_1)}  \right)
- \left( ^2A_{3} \frac{dz}{(z-p_2)^{3}}+ ^2A_{2} \frac{dz}{(z-p_2)^{2}}  + ^2A_{1} \frac{dz}{(z-p_2)}\right).\end{equation}
Using the examples given in this chapter, consider the following:
\begin{itemize}
\item[i.] Describe the sets of anti-Stokes directions and  identify whether the Stokes matrices are upper or lower triangular;
 \item[ii.] Give a description of the tentacle as in Figure \ref{wild10}.
\item[iii.] For Higgs bundles with one pole of order 4 as in Remark \ref{rem4}, show that the system has 6 anti-stokes directions and give a description of them.
\end{itemize}

 \end{problem}


\section{The geometry of the Hitchin fibration}

\epigraph{\textit{%
Integrability of a system of differential equations should manifest itself through
some generally recognizable features:\\
$\star$ the existence of many conserved quantities;\\
$\star$ the presence of algebraic geometry;\\
$\star$ the ability to give explicit solutions.\\
These guidelines should be interpreted
in a very broad sense.''}}{Nigel Hitchin}

The moduli space $\mathcal{M}_{ G_\C}$ of $G_\C$-Higgs bundles can be naturally studied through the Hitchin fibration, as introduced in \cite{N2}. We shall dedicate this chapter to the study of this fibration and the description of some of its geometric properties. 

\subsection{The Hitchin fibration and the Teichm\"uller component.}

 This Hitchin fibration is defined through a homogeneous basis   $p_{i}$,  for $i=1,\ldots, k$ for the algebra of invariant polynomials of the Lie algebra  $\mathfrak{g}_{c}$ of $ G_\C$, whose degrees shall be denoted by $d_{i}$. Then, the {\em Hitchin fibration} is given by
\begin{align}
 h~:~ \mathcal{M}_{ G_\C}&\longrightarrow\mathcal{A}_{ G_\C}:=\bigoplus_{i=1}^{k}H^{0}(\Sigma,K^{d_{i}}),\\
 (E,\Phi)&\mapsto (p_{1}(\Phi), \ldots, p_{k}(\Phi)).
\end{align}
The map $h$ is referred to as the {\em Hitchin~map}, and is a proper map for any choice of basis (see \cite[Section 4]{N2} for details). Furthermore, the Hitchin base $\mathcal{A}_{ G_\C}$ always satisfies
\[\text{dim} \mathcal{A}_{ G_\C} =\text{dim}\mathcal{M}_{ G_\C}/2,\]
making the Higgs bundle moduli space into an integrable system.
 
By considering different homogenous basis of invariant polynomials, one can make different aspects of the geometry of the Hitchin fibration explicit. In particular, one may sometimes ask whether certain property of the fibration (or of a fibre) is preserved under a change of basis, and we shall see in later sections examples of this. Across these notes,  for classical Higgs bundles $(E,\Phi)$ of rank $n$ we shall often take the basis given by $\textrm{tr}(\Phi^i)$ for $1\leq i\leq n$, or given  by the coefficients of the characteristic polynomial of the Higgs field $\Phi$. 

\begin{remark}
Through changes of basis, one obtains correspondences between different expressions of the Hitchin fibration for the same group. For instance, the coefficients of the characteristic polynomial can be expressed in terms of traces as follows:
\begin{equation}p(x)=x^n + \textrm{tr} (\Phi) x^{n-1}+ \sum_{k=2}^{n-1} x^{n-k}~(-1)^k ~\textrm{Tr}(\Lambda^k \Phi) +(-1)^n~\textrm{det}(\Phi)\label{uno1} \end{equation}
where $\textrm{Tr}(\Lambda^k \Phi) $ is the trace of the $k$th exterior power of $\Phi$, with dimension $\frac{n!}{k!(n-k)!}$
\end{remark}

The  \textit{Teichm\"{u}ller component} \cite{N5}, or so-called \textit{Hitchin section} of the Hitchin fibration,  is induced by the   \textrm{Kostant slice} $\mathcal{K}$ of $\mathfrak{g}_{\C}$, which is the space given by
\begin{equation}\mathcal{K}:=\{f\in \mathfrak{g}_{\C}~|~ f= f_{0} +a_{1}e_{1}+ a_{2}e_{2}+\ldots+ a_{r}e_{r}\},\label{decf}\end{equation}
for $e_{i}$ the highest weight vectors of $\mathfrak{g}_{\C}$ and $a_{i}$  complex numbers, and $f_{0}$ a nilpotent element of  a \textrm{three-dimensional subalgebra
 $\mathfrak{s}$} of $\mathfrak{g}_{\C}$. Recall that this is a subalgebra of $\mathfrak{g}_{\C}$ generated by a semisimple element $h_{0}$, and nilpotent elements
 $e_{0}$ and $f_{0}$ of  $\mathfrak{g}_{\C}$ satisfying the relations
\begin{equation}[h_{0},e_{0}]=e_{0}~;~ [h_{0},f_{0}]=-f_{0}~;~ [e_{0},f_{0}]=h_{0}.\end{equation}

From the decomposition of $\mathfrak{g}_{\C}$ into eigenspaces of $\textrm{ad}h_{0}$,   consider the vector bundle  
\begin{equation}
 E= \textrm{ad}P\otimes \mathbb{C}=\bigoplus_{j=-k_{r}}^{k_{r}}\mathfrak{g}_{j}\otimes K^{j},
\end{equation}
where $k_{r}=\textrm{max}\{k_{i}\}_{i}$ is the maximal exponent of $\mathfrak{g}_{\C}$.  This is the adjoint bundle of $\mathfrak{g}_{\C}$ associated to the principal $\textrm{Ad}G_{\C}$-bundle $P=P_{1}\times_{i} G$, where $P_{1}$ is the holomorphic principal $SL(2,\mathbb{C})$-bundle associated to $K^{-1/2}\oplus K^{1/2}$, and   the inclusion $i:SL(2,\mathbb{C})\hookrightarrow \textrm{Ad}G_{\C}$ corresponds to the principal three dimensional subalgebra $\mathfrak{s}_{0}$. Although $P_{1}$ involves a choice of square root $K^{1/2}$, the bundle $E$ is independent of it.  The $\textrm{Ad}G_{\C}$-Higgs bundle $(E,\Phi)$, for
\begin{equation}\Phi= f_{0} +a_{1}e_{1}+ a_{2}e_{2}+\ldots+ a_{r}e_{r}\label{decp}\end{equation}
and $a_{i}\in H^{0}(\Sigma,K^{k_{i}+1})$, is stable. Note that $(\textrm{ad}h_{0})f_{0}=-f_{0}$ and thus we may regard $f_{0}$ as a section of $(\mathfrak{g}_{-1}\otimes K^{-1})\otimes K$. Furthermore, the highest weight vectors $e_{j}\in \mathfrak{g}_{k_{j}}$ and thus $a_{j}e_{j}$ is a section of $\mathfrak{g}_{k_{j}}\otimes K^{k_{j}+1}$, making $\Phi$ a well defined holomorphic section of $E\otimes K$. 
Since $p_{i}(\Phi)=a_{i}$, the above construction defines a section 
\begin{equation}s: (a_{1},\ldots,a_{r})\mapsto (E,\Phi)\label{tei}\end{equation} of $\mathcal{M}_{\textrm{Ad}G_{\C}}$ whose image is the \textit{Teichm\"{u}ller component}. This component defines an origin in the smooth fibres of $h$, and the Hitchin fibration becomes a fibration of abelian varieties (for more details the reader should refer to \cite{N2} where Hitchin introduced the  \textit{Teichm\"{u}ller component})
 The following result (\cite[Theorem 7.5]{N5}) relates the Teichm\"{u}ller component to the space of representations of the split real form of $\textrm{Ad}G_{\C}$:

\begin{theorem} 
 The section $s$ of $\mathcal{M}_{\textrm{Ad}G_{\C}}$ defines a smooth connected component of the moduli space of reductive representations of $\pi_{1}(\Sigma)$ into the split real form of $\textrm{Ad}G_{\C}$. 
\end{theorem}

\subsection{The regular fibres of the Hitchin fibration.}

We shall first consider the regular fibres of the Hitchin fibration, which are those fibres over regular -- generic -- points in the Hitchin base. Formally, we shall be considering fibres over points $a\in \mathcal{A}_{ G_\C}$ defining \textit{smooth} spectral curves $S$.  As before, we consider $K=T^*\Sigma$, and $X=\textrm{Tot}(K)$ its total space with projection $\pi:X\rightarrow \Sigma$. We shall denote by $\eta$ the tautological section of the pull back $\pi^{*}K$ on $X$. Abusing notation we denote with the same symbols the sections  of powers $K^{i}$ on $\Sigma$ and their pull backs to $X$ (for more details see the equations in \cite{BNR} where the distinction is made when defining spectral data).

 The characteristic polynomial of a Higgs bundle $(E,\Phi)$ in a generic fibre $h^{-1}(a)$ defines a smooth curve $\pi:S_a\rightarrow \Sigma$ in $X$, the {\em spectral curve} of $\Phi$, whose equation is
\begin{equation}
\text{det}(\eta - \Phi)= \eta^{n}+a_{1}\eta^{n-1}+a_{2}\eta^{n-2}+\ldots + a_{n-1}\eta+a_{n}=0,\label{smootheq}
\end{equation}
for $a_{i}\in H^{0}(\Sigma,K^{i})$ (for simplicity, we will  drop the subscript $a$ of $S_a$).  By the adjunction formula on $X$ (see e.g. \cite{harris}), The cotangent bundle of $\Sigma$ is a symplectic manifold and hence has trivial canonical bundle, so $K_S\otimes \pi^*K^{-m}$ is trivial and  $K_S= \pi^*K^{n}$. Taking degrees in both sides gives the genus 
 \begin{equation}g_{S}=1+n^{2}(g-1).\label{genusSclassical}\end{equation}

 The \textit{spectral data} associated to a Higgs bundle $(E,\Phi)$ provides a geometric description of the fibres of the Hitchin fibration as abelian varieties, and  is then given by:
\begin{itemize}
\item the spectral curve $\pi:S\rightarrow \Sigma$, and
\item a vector bundle   on $S$, satisfying appropriate conditions reflecting the nature of   $G_\C$.
\end{itemize}
 In the case of classical Higgs bundles, the smooth fibres are $\textrm{Jac}(S)$. For $\eta$ the tautological section of $\pi^*K$, one   recovers   $(E,\Phi)$ from the curve $(S, L\in \textrm{Jac}(S))$ by taking $E=\pi_*L$ and $\Phi=\pi_* \eta$. 
Indeed, recall that by the definition of direct image sheaf, given an open set $\mathcal{U}\subset \Sigma$, one has that 
$  H^{0}(\pi^{-1}(\mathcal{U}),L)= H^{0}(\mathcal{U}, \pi_{*}L), $
and multiplication by  $\eta$ induces the map
\begin{equation}
 H^{0}(\pi^{-1}(\mathcal{U}),L)\xrightarrow{\eta} H^{0}(\pi^{-1}(\mathcal{U}),L\otimes\pi^{*}K).\nonumber
\end{equation}
By considering the direct image of the map, one obtains \begin{equation}
 \Phi:    \pi_{*}L\rightarrow  \pi_{*}L \otimes K,\nonumber
\end{equation}
giving a Higgs field $\Phi\in H^{0}(\Sigma,\text{End}E \otimes K)$ for $E:=\pi_{*}L$.
 Moreover, the Higgs field satisfies its characteristic equation, which by construction is given by the equation 
$\eta^{n}+a_{1}\eta^{n-1}+a_{2}\eta^{n-2}+\ldots + a_{n-1}\eta+a_{n}=0.$
Furthermore, since $S$ is irreducible, from Remark \ref{invariantsubbundle} there are no invariant subbundles of the Higgs field, making the induced Higgs bundle $(E,\Phi)$  stable.
Conversely, let $(E,\Phi)$ be a classical Higgs bundle. The characteristic polynomial is given by
\eqref{smootheq} 
and its coefficients  define the $spectral~curve$ $S$ in the total space $X$.

From \cite[Proposition 3.6]{BNR}, there is a bijective correspondence between Higgs bundles $(E,\Phi)$ and the line bundles $L$ on the spectral curve $S$ described previously. This correspondence identifies the fibre of the Hitchin map with the Picard variety of line bundles of the appropriate degree.  By tensoring the line bundles $L$ with a chosen line bundle of degree $-\text{deg}(L)$, one obtains a point in the Jacobian $\text{Jac}(S)$, the abelian variety of line bundles of  degree zero on $S$, which has dimension $g_{S}$ as in (\ref{genusSclassical}). In particular, the Jacobian variety is the connected component of the identity in the Picard group $H^{1}(S,\mathcal{O}^{*}_{S})$.
Thus, the fibre of the classical Hitchin fibration $h :\mathcal{M}\rightarrow \mathcal{A} $ is isomorphic to the Jacobian of the spectral curve $S$. For more details, the reader should refer for example to \cite[Section 2]{N3}.

\begin{example}\label{newexample}
In the case of a classical rank 2 Higgs bundle $(E,\Phi)$, the characteristic polynomial of $\Phi$ defines a spectral curve $\pi:S \rightarrow \Sigma$. This is a 2-fold cover of $\Sigma$ in the total space of $K$, and has  equation
$\eta^{2}+a_{2}=0$,  
for $a_{2}$ a quadratic differential  and $\eta$ the tautological section of $\pi^{*}K$. By \cite[Remark 3.5]{BNR} the curve is smooth when $a_{2}$ has simple zeros, and in this case the ramification points are  given by the divisor of $a_{2}$. For $z$ a  local coordinate near a  ramification point, the covering is given by
$ z\mapsto z^{2}:=w.$ In a neighbourhood of $z=0$, a section of the line bundle $M$ looks like
$f(w)=f_{0}(w)+zf_{1}(w).$
Since the Higgs field is obtained via multiplication by $\eta$, one has
\begin{equation}
\Phi (f_{0}(w)+zf_{1}(w)) = w f_{1}(w)+ z f_{0}(w),
\end{equation}
and thus  a local form of the Higgs field $\Phi$ is given by 
\[\Phi=\left(\begin{array}
              {cc}
0&w\\
1&0
             \end{array}
\right).\]
                              

\end{example}

By considering additional conditions on the Higgs field $\Sigma$ (e.g., requiring it to have a $U(p,q)$-structure as in Example \ref{expq}) the characteristic polynomial, and thus the curve it defines,  acquires further structure, and examples of these will be described below. Hence, when $ G_\C\subset GL(n,\mathbb{C})$, the spectral data of a $G_\C$-Higgs bundle is given by the spectral data of the pair as a classical Higgs bundle, satisfying extra conditions. An example of this is how the spectral data for real $G$-Higgs bundles can be expressed when $G$ is the split real form of the complex Lie group  $ G_\C$.

\begin{theorem}[\cite{thesis} Theorem 4.12] The spectral data for $G$-Higgs bundles  when $G$ is the split real form of the complex Lie group  $ G_\C$ is given by the $G_\C$ spectral data of order two. In other words, 
the intersection of the subspace of the Higgs bundle moduli space $\mathcal{M}_{G_{\C}}$ corresponding to the split real form of $\mathfrak{g}_{\C}$  with the smooth fibres of the Hitchin fibration
\[h:~\mathcal{M}_{G_{\C}}\rightarrow \mathcal{A}_{G_{\C}},\]
 is given by the elements of order two in those fibres.
\end{theorem}

Through the above theorem, one can certain real Higgs bundles as a covering space in the Hitchin fibration, and thus their geometry can be studied through the properties of covering spaces -- for instance, through the monodromy action for the fibration. To define the monodromy action for the Hitchin fibration, consider a generic fibration $p:Y \rightarrow B$ which is locally trivial, i.e. for any point $b\in B$ there is an open neighbourhood $U_{b}\in B$ such that $p^{-1}(U_{b})\cong U_{b}\times Y_{b}$ where $Y_{b}$ denotes the fibre at $b$.   The $n$th homologies of the fibres $Y_{b}$ form a locally trivial vector bundle over $B$, which we denote $\mathcal{H}_{n}(B)$. This bundle carries a canonical flat connection, the \textit{Gauss-Manin connection}.

To define this connection we  identify the fibres of $\mathcal{H}_{n}(B)$ at nearby points $b_{1},b_{2}\in B$, i.e. $H_{n}(Y_{b_{1}})$ and $H_{n}(Y_{b_{2}})$. Consider $N\subset B$ a contractible open set which includes $b_{1}$ and $b_{2}$.  The inclusion of the fibres $Y_{b_{1}}\hookrightarrow p^{-1}(N)$ and $Y_{b_{2}}\hookrightarrow p^{-1}(N)$ are homotopy equivalences,  and hence we obtain an isomorphism between the homology of a fibre over a point in a contractible open set  $N$ and $H_n(p^{-1}(N))$: \[H_{n}(Y_{b_{1}})\cong H_{n}(p^{-1}(N))\cong H_{n}(Y_{b_{2}}).\]
This means that the vector bundle $\mathcal{H}_n(B)$ over $B$  has a flat connection, the Gauss-Manin connection.   The \textit{monodromy} of $\mathcal{H}_n(B)$   is the holonomy of this connection, i.e. a homomorphism $\pi_{1}(B)\rightarrow \textrm{Aut}(H_{n}(Y))$ as an action of $\pi_{1}(B)$ on $H_{n}(Y)$.  By applying these results  to the fibration $\M \rightarrow \A$  one has that the   Gauss-Manin connection on the cohomology of the fibres of $h:\M\rightarrow \A$ defines the monodromy action for the Hitchin fibration. This monodromy action as been studied for some groups in \cite{mono,thesis,mono1,mono2}.

\subsubsection{The Hitchin fibration and $SL(n,\C)$-Higgs bundles.}
 To give an example of the conditions needed, we shall consider  $ G_\C=SL(n,\mathbb{C})$. Here, through Definition \ref{defHiggs} an $SL(n,\mathbb{C})$-Higgs bundle is a classical Higgs bundle $(E,\Phi)$ where the rank $n$ vector bundle $E$ has trivial determinant and the Higgs field has zero trace.  
A basis for the invariant polynomials on the Lie algebra $\mathfrak{sl}(n,\mathbb{C})$ is given by the coefficients of the characteristic polynomial of a trace-free matrix $A\in \mathfrak{sl}(n,\mathbb{C})$, and thus the spectral curve $\pi: S\rightarrow \Sigma$ defined by the Higgs field has equation \eqref{smootheq} with $a_1=0$.
 In this case the generic fibres of the Hitchin fibration are given by the subset of $\text{Jac}(S)$  of line bundles $L$ on $S$ for which $\pi_{*}L=E$ and $\Lambda^{n}\pi_{*}L$ is trivial: the generic fibre  of the $SL(n,\mathbb{C})$ Hitchin fibration is biholomorphically equivalent to the Prym$(S,\Sigma)$, for $S$    the  spectral curve   defined as in (\ref{smootheq}) with the coefficient $a_1=0$. 
 
 In order to give the flavour of some of the geometry arising through the Hitchin fibration, we shall consider next the case of $n=2$ and let $\mathcal{A}$ be the $SL(2,\C)$ Hitchin base and by $\mathcal{M}$ the moduli space of $SL(2,\C)$-Higgs bundles. 
For any $\omega\in \mathcal{A}-\{0\}$ it is shown in \cite[Theorem 8.1]{goand} that the fibre $\mathcal{M}_{\omega}$ is connected.  For any isomorphism class of  $(E,\Phi)$ in $\mathcal{M}$, one may consider the zero set of its characteristic polynomial 
\[\textrm{det}(\Phi-\eta I)= \eta^{2}+\omega=0,\]
where $\omega=\textrm{det}(\Phi)\in \mathcal{A}$. This defines a spectral curve $\rho:S\rightarrow \Sigma$ in the total space $X$ on $K$, for $\eta$  the tautological section of the pull back of $K$ on $X$. We shall denote by $\M$ the regular fibres of the Hitchin map $h$, and let $\A$ be the regular locus of the base, which is given by quadratic differentials with simple zeros. Note that the curve $S$ is non-singular over the regular locus $\A$, and the ramification points are given by the intersection of $S$ with the zero section.
The curve $S$ has a  natural involution $\sigma(\eta)=-\sigma$ and thus we can define the Prym variety $\textrm{Prym}(S,\Sigma)$ as the set of line bundles $M\in \textrm{Jac}(S)$ which satisfy
\[\sigma^{*} M\cong M^{*}.\]
In particular, this definition is consistent with the one given for $n>2$ by means of the Norm map. 

Since $S$ is a 2-fold cover in the total space $X$ of $K$, the direct image of the trivial bundle $\mathcal{O}$ in $\textrm{Prym}(S,\Sigma)$ is given by $\pi_{*}\mathcal{O}=\mathcal{O}\oplus K^{-1}$ (e.g. see \cite[Remark 3.1]{BNR}). Moreover, from the natural involution $\sigma$,   the sections of $\pi_{*}\mathcal{O}$ can be separated into invariant and anti-invariant ones. As $\mathcal{O}$ corresponds to the invariant sections, and $\Lambda^{2}\pi_{*}\mathcal{O}\cong K^{-1}$, necessarily $\pi_{*}\mathcal{O}=\mathcal{O}\oplus K^{-1}$.

The \textit{Hitchin section} as introduced in \eqref{tei} of the Hitchin fibration can be obtained naturally from  the line bundle $L=\pi^{*}K^{1/2}$ on $S$, for which one has
\[\pi_{*}\pi^{*}K^{1/2}= K^{1/2}\otimes \pi_{*}\mathcal{O}=K^{1/2}\oplus K^{-1/2}.\]
Hence, the line bundle $\mathcal{O}\in \textrm{Prym}(S,\Sigma)$ has an associated Higgs bundle given by $(K^{1/2}\oplus K^{-1/2}, \Phi_{\omega})$, where the Higgs  field $\Phi_{\omega}$ defined as in Example \ref{exa} and Example {newexample} by 
\[\Phi_{\omega}=\left(\begin{array}{cc}
              0&\omega\\1&0
             \end{array}
\right),~{~\rm~for~}~ \omega\in H^{2}(\Sigma,K^{2}).\]

   For generic $ G_\C$, a  description of the fibres can be obtained by means of Cameral covers  \cite{donagi1} (see also \cite{Don93}) which is equivalent to the one given in the next sections for classical Lie groups. For a comprehensive description of Cameral covers for real Higgs bundles, the reader may refer to \cite{ana1,peon}.

\subsubsection*{The split real forms.} In order to study Higgs bundles for the split real form $SL(2,\R)$, one should remember that their spectral data corresponds to  elements of order 2 in the abelian varieties $\textrm{Prym}(S,\Sigma)$ giving the generic fibres of the Hitchin fibration. Moreover,  in this case the monodromy is generated by the action of $\pi_{1}(\A)$ on $H^{1}( \textrm{Prym}(S,\Sigma),\mathbb{Z})$, a space which in turn can be shown to give the set $P[2]$ of  torsion two points of the fibre (this can be generalized to higher rank, and the reader should refer to \cite{mono1} for a more detailed study). 

 In what follows we shall see how the space  $H^{1}( \textrm{Prym}(S,\Sigma),\mathbb{Z})$ is isomorphic to the torsion two points $P[2]$ in the generic fibre $\textrm{Prym}(S,\Sigma)$. Since a generic fibre of $h$ is given by the abelian variety $\textrm{Prym}(S,\Sigma)$, i.e., by the quotient of a complex vector space $V$ by some lattice $\wedge$, one has an associated exact sequence of homology groups
\[\ldots \rightarrow \pi_{1}(\wedge)\rightarrow \pi_{1}(V)\rightarrow \pi_{1}(\textrm{Prym}(S,\Sigma))\rightarrow \pi_{0}(\wedge)\rightarrow \ldots .\]
Hence, there is a short exact sequence
$0\rightarrow \pi_{1}(\textrm{Prym}(S,\Sigma))\rightarrow \wedge \rightarrow 0,$
from where $\wedge\cong \pi_{1}(\textrm{Prym}(S,\Sigma))$. Therefore, $\pi_{1}(\textrm{Prym}(S,\Sigma))$ is an abelian group, i.e., \[\pi_{1}(\textrm{Prym}(S,\Sigma))\cong H_{1}(\textrm{Prym}(S,\Sigma),\mathbb{Z}).\]
We shall denote by $P[2]$ the elements of order 2 in $\textrm{Prym }(S,\Sigma)$, which are equivalent classes in $V$ of points $x$ such that $2x \in \wedge$.
 Then, $P[2]$ is given by $\frac{1}{2}\wedge$ modulo $\wedge$, and as $\wedge$ is torsion free, \[P[2]\cong \wedge/2\wedge\cong H_{1}(\textrm{Prym}(S,\Sigma),\mathbb{Z}_{2}).\]
Moreover, $H^{1}(\textrm{Prym}(S,\Sigma),\mathbb{Z}_{2})\cong \textrm{Hom}(H_{1}(\textrm{Prym}(S,\Sigma),\mathbb{Z}),\mathbb{Z}_{2})$ and thus
\[H^{1}(\textrm{Prym}(S,\Sigma),\mathbb{Z}_{2})\cong \textrm{Hom}(\wedge,\mathbb{Z}_{2})\cong \wedge/2\wedge\cong P[2].\]
Hence, the covering space $P[2]$  is determined by the action of $\pi_{1}(\A)$ on the first cohomology of the fibres with $\mathbb{Z}_{2}$ coefficients. This action is studied in \cite{mono} for $n=2$, and in \cite{mono1,mono2} for $n>2$. 

\subsubsection*{Real forms with signature}  The locus of $SU(p,p+q)$-Higgs bundles   is   fixed by an  involution  $\Theta_{SU(p,p+q)}:(E,\Phi)\mapsto (E,-\Phi)$ on $SL(2p+q,\mathbb{C})$-Higgs bundles corresponding to bundles $E$ which have an automorphism conjugate to $I_{p,p+q}$ sending $\Phi$ to $-\Phi$, and  whose $\pm 1$ eigenspaces have dimensions $p$ and $p+q$ (see also Example \ref{expq}).  The involution $-\sigma$ acts trivially on the polynomials of even degree. In this case, the characteristic polynomial is given by \begin{equation}\text{det}(\eta-\Phi)= \eta^q(\eta^{2p}+a_{2}\eta^{2p-2}+\ldots+a_{2p-2}\eta^2+a_{2p}),\label{eqpq}\end{equation}
 defines a spectral curve $\pi:S\rightarrow \Sigma$ as
\[S=\{0=\eta^{2p}+a_{2}\eta^{2p-2}+\ldots+a_{2p-2}\eta^2+a_{2p}\}\subset \textrm{Tot}(K),\]
where $\eta$ is the tautological section of $\pi^*K$ and $a_i\in H^0(\Sigma,K^i)$.

Whilst the spectral data is not known for $q \neq 0,1$, in the case of $q=0$ it has been described in \cite[Chapter 6]{thesis} and \cite{umm} by looking at $U(p,p)$-Higgs bundles $(W_1\oplus W_2, \Phi)$, which when satisfying $\Lambda^p W_1\cong \Lambda^q W_2^*$ correspond to $SU(p,p)$-Higgs bundles (the case of $q=1$ is studied in \cite{peon} through Cameral data). Moreover, for $q=0$, the curve $S$ is a $2p$-fold cover of the Riemann surface $\Sigma$, ramified over the $4p(g-1)$ zeros of $a_{2p}$,  and has a natural involution \begin{equation}\label{invopq}\sigma:\eta\mapsto -\eta\end{equation} which has as fixed points the ramification points of the cover. The quotient of $S$ by the action of $\sigma$ defines a $p$-fold cover $\overline{\rho}:\overline{S}\rightarrow \Sigma$ in the total space of $K^{2}$, whose equation is  given by
$
\xi^{p}+a_{1}\xi^{p-1}+\ldots+a_{p-1}\xi+a_{p}=0,
 \label{curvess}
$
where $\xi=\eta^{2}$ is the tautological section of $\overline{\rho}^{*}K^{2}$. Since $\overline{S}$ is the quotient of a smooth curve, it is also smooth. We let $\pi:S\rightarrow \overline{S}$ be the double cover given by the above quotient:
 \begin{equation}
 \xymatrix{S\ar[rr]_{\pi}^{2:1}\ar[dr]^{\rho}_{2p:1}&&\overline{S}\ar[dl]^{p:1}_{\overline{\rho}}.\\
&\Sigma&}
\end{equation}
 We shall denote by $g_{S}$ and $g_{\overline{S}}$ the genus of  $S$ and $\overline{S}$, respectively.
Since the cotangent bundle has trivial canonical bundle, and the canonical bundle of $K^2$ is $\bar \rho^*K^{-1}$, the adjunction formula gives 
  $K_{S}\cong \rho^{*}K^{2p}$ and  $K_{\bar S}\cong  \bar \rho^{*} K^{2p} \otimes \bar \rho^*K^{-1}$. Thus, \begin{align}g_{S}&=4p^{2}(g-1)+1,\\
 g_{\overline{S}}&= (2p^{2}-p)(g-1)+1.\end{align}
The spectral data for $SU(p,p)$-Higgs bundles can then be deduced from the spectral data of $U(p,p)$-Higgs bundles described as follows:
\begin{theorem}[\cite{umm}] Given a $U(p,p)$-Higgs bundle with non-singular spectral curve one can construct a pair $(S,M)$ where 
\begin{enumerate}
\item[(a)] \label{in1} the curve $\rho:S\rightarrow \Sigma$ is an irreducible non-singular $2p$-fold cover of $\Sigma$  given by the equation
$\eta^{2p}+a_{1}\eta^{2p-2}+\ldots+a_{p-1}\eta^{2}+a_{p}=0,\nonumber$
in the total space of $K$, where $a_{i}\in H^{0}(\Sigma, K^{2i})$, and $\eta$ is the tautological section of $\rho^{*}K$. The curve $S$ has an involution $\sigma$ acting by $\sigma(\eta)=-\eta$;

\item[(b)] \label{in2} $M$ is  a line bundle on $S$ such that $\sigma^{*}M\cong M$. 
 
\end{enumerate}
 Conversely, given a pair $(S,M)$ satisfying (a) and (b), there is an associated stable $U(p,p)$-Higgs bundle whose spectral curve is as in (a).
\end{theorem}

From the structure of the characteristic polynomials of $SU(p,p+q)$-Higgs bundles described in \eqref{eqpq} one can see that for $q\geq 1$ generically the polynomial defines a reducible curve, and thus all of the $SU(p,p+q)$-Higgs bundles lie in the fibres over the discriminant locus of the Hitchin fibration. Most other non-split real forms also give Higgs bundles which lie over the non-regular locus of the Hitchin fibration, as shown in \cite{thesis}.

More generically, one may ask whether a given set of Higgs bundles can be ever expressed as elements over the regular locus of the Hitchin fibration. For the case of rank 2 Higgs bundles, Wentworth and Wolf proved in \cite{went} that for  every non-elementary representation of a surface group into $SL(2,\C)$ there is a Riemann surface structure such that the Higgs bundle associated to the representation lies outside the discriminant locus of the Hitchin fibration For $n\geq 3$, and for any other group, the question remains open: indeed, from \cite[Remark 3.iv.]{went}  there will be  obstructions in any generalization of their theorem, and some of these will come
from other real forms of $SL(n,\C)$ which always lie in the discriminant locus.  In particular, their work involves a result of Gallo-Kapovich-Marden which would need to be generalized for higher rank.

\subsection{The singular locus of the Hitchin fibration.}
Across these notes we have said that  a point $\omega \in \mathcal{A}_{G_\C}$ in the Hitchin base is    regular when the spectra curve $S$ it defines is smooth -- equivalently, we said it is singular when the corresponding spectral curve  defined through the characteristic polynomial of the Higgs field \eqref{smootheq} is singular. There are, however, more elaborated and useful ways of stratifying the Hitchin base, as for instance those used in \cite{ngo} (see \cite{lucas} for further descriptions of those stratifications). However, we shall restrict ourselves here to the simpler definition of singular point and make some comments on some of the arising geometry and open problems.  

 The most singular fiber is the fibre over $\mathbf{0} \in \mathcal{A}_{G_\C}$,   named  \textit{the nilpotent cone} by Laumon  \cite{laumon}, to emphasise the analogy with the nilpotent cone in Lie algebra.
 This fibre has been studied from different perspectives, one of which is via the moment map $\mu$ of the $S^1$ action 
$(E, \Phi) \to (E, \mathrm{e}^{i\theta} \Phi)$ 
 as done in \cite{N1,hauseldiss}. In particular, note that given a stable bundle $E$, 
take $\Phi=0$, then the Higgs bundle $(E, 0)$ is stable
and trivially fixed by the $S^1$-action.

As shown in  \cite[Theorem 5.2]{Hausel:1998aa}, 
 the nilpotent cone is preserved by the flow by $\mu$ and  it encodes the topology of the moduli space  since  points
of $\mathcal{M}_{G_\C}$ flow towards the nilpotent cone. For a overview of work done concerning singular Higgs bundles and Higgs bundles over the singular locus of the Hitchin fibration the reader should refer to \cite{SIGMA}. In particular, the nilpotent cone has primarily been studied for $SL(n,\C)$ and $GL(n,\C)$, and much of its geometry remains unknown for the moduli spaces of $G_\C$-Higgs bundles. For $SL(n,\C)$ and $GL(n,\C)$-Higgs bundles,
the irreducible components of the nilpotent cone are labeled by connected components
of the fixed point set of the $S^1$ action. The nilpotent cone can be also considered for other types of Higgs bundles, and we shall make some comments in Section~\ref{polii} about the nilpotent cone appearing through hyperpolygons. 
Among the components of the nilpotent cone for classical Higgs bundles is the moduli space $\mathcal{N}$
of semistable bundles.

Other fibers over the singular locus of the Hitchin fibration have been the subject of more recent research  (e.g., see  \cite{lucas,ortlau,goand,cayley,nonabelian,mas10}). Since we will return to singular fibres when studying branes of Higgs bundles, we shall conclude this chapter with   few short comments: 

\begin{itemize}
\item In the case of $GL(n,\C)$-Higgs bundles we mentioned that for generic points of the Hitchin base,   the corresponding fiber   can be identified  with the Jacobian $\mathrm{Jac}(S)$ of the spectral curve $S$.  When the spectral curve $S$ is not smooth but is integral, the corresponding fiber
 is seen to be the \emph{compactified Jacobian}    \cite{BNR, schaub,simpson,simpson88} (see also \cite[Fact 10.3]{melo} for a clear explanation).
  The compactified Jacobian $\overline{\mathrm{Jac}}(S)$ is the moduli space of
 all torsion-free rank-1 sheaves on $S$, where the usual ``locally-free''
 condition is missing. Intuitively,  the compactified Jacobian can be obtained by considering a path of smooth curves $S_t$ approaching a singular curve $S_0$ which is the base point of the nilpotent cone; since the limit of $\mathrm{Jac}(S_t)$ does not depend on the choice of smooth family  \cite{igusa},  this limit is $\overline{\mathrm{Jac}}(S)$. \item For $GL(n,\C)$-Higgs bundles whose spectral curve $S$ is not integral,   the fine moduli space needs to be considered.

\item For $SL(2,\C)$-Higgs bundles,  connectedness of the fiber of $\mathcal{M}_{SL(2,\C)}$ when $S$ is irreducible and has only simple nodes has been considered in \cite[\S5.2.2]{frenkelwitten}, whilst a full description of the singular fibers is given in   \cite{goand}.

 \end{itemize}

\subsection{Problem set II.}
This problem set is meant to make the  reader used to \textit{spectral data} for complex Higgs bundles, and understand how it can be expressed in an abelian manner. Moreover, it should provide an introduction to ways in which one can study singular fibres of the Hitchin fibration. They provide a mixture of problems from the lectures as well as some additional ones.  

\begin{problem}\textbf{Spectral curves.} Through its characteristic polynomial, a Higgs field defines a natural spectral curve in the total space of the canonical bundle $K$. Moreover,  the extra conditions satisfied by the Higgs field are reflected, in particular, in the properties of this curve. 
\begin{itemize}
\item[i.]  For which complex Lie groups among the types $A,B,C,D$, the spectral curves defined by the corresponding Higgs field carry  an involution?
\item[ii.] Considering low rank groups, find the conditions on a Higgs bundle $(E,\Phi)$ for the Higgs field to define a spectral curve which has an order 4 automorphism (e.g. see \cite{auto1}).   
\item[iii.] By considering properties of characteristic polynomials   from a Linear Algebra perspective, give a description of two particular types of Higgs bundles $(E,\Phi)$ for which the spectral curve $S=\{(\det \Phi-\lambda)=0\}$ has equation $\{q^3(\lambda)=0\}$ for some irreducible polynomial $q(\lambda)$. 
\item[iv.] Following on \textbf{Problem 6.iii.}, for which Higgs bundles one has  \[(\det \Phi-\lambda)=q^n(\lambda)\] for some irreducible polynomial $q(\lambda)$?
\end{itemize}
\end{problem}
\begin{problem} \textbf{Generic fibres.} We have mentioned that by adding conditions to a Higgs pair $(E,\Phi)$, one can obtain a description of a principal $G_\C$-Higgs bundle. Equivalently, one may add conditions to the classical spectral curve $S$, and data on it $\textrm{Jac}(S)$ to recover principal Higgs bundles. 
\begin{itemize}
\item[i.] Use Grothendieck-Riemann-Roch to  $\text{deg}(E)=\text{deg}(L) + (n^{2}-n)(1 - g)$ when considering spectral data for classical Higgs bundles.  
\item[ii.] Show that when considering $SL(n,\C)$-Higgs bundles, the condition \[\Lambda
^n E\cong \CO\] is equivalent to requiring the spectral line bundle $L$ to be in $\textrm{Prym}(S,\Sigma)$.  
%
%
\item[iii.] Show that the smooth $SO(2n,\mathbb{C})$-Hitchin  fibres    are  $\text{Prym}(\hat{S},\hat{S}/\hat{\sigma})$, for $\hat S$ the desingularization of the curve $S$ associated to the regular base point.
\item[iv.] Show that when considering generic $SO(2n+1,\C)$-Higgs bundles, the spectral curve $S$ defined through the characteristic polynomial is never smooth. 
\end{itemize}

\end{problem}
\begin{problem}  \textbf{The Hitchin fibration}. When considering a subspace of Higgs bundles which has finite intersection with the regular locus, one can study its geometry and topology though the monodromy action as well as through other involutions.

\begin{itemize}
\item[i.]There is a natural involution $\sigma_{s}$ on the Lie algebra $\mathfrak{g}_{\C}$  given by
\begin{equation}\sigma_{s}( [\textrm{ad } f_{0}]^{n}e_{i})=(-1)^{n+1}[\textrm{ad } f_{0}]^{n}e_{i}.\label{involution}\end{equation}
This    automorphism   is uniquely defined by  
$\sigma_{s}(e_{i})=-e_{i}$ and $\sigma_{s}(f_{0})=-f_{0}.$ Show that  $-\sigma_{s}$  on $\mathfrak{g}_{\C}$   acts trivially on the ring of invariant polynomials of the Lie algebra $\mathfrak{g}_{\C}$.

\item[ii.] Show that for non-split real forms $G$, the characteristic polynomial of a $G$-Higgs field defines a reducible curve.
\item[iii.] \textbf{((*))} Generalize  \cite{bhosle3} to define generalized   parabolic $G_\C$-Higgs bundles on $X$, as well as a Hitchin fibration.

\end{itemize}

\end{problem}
\begin{problem}  \textbf{Singular fibres}.  Consider    $SU(p,p+q)$-Higgs bundles $(V\oplus W,\Phi)$, whose characteristic polynomial is   \[\text{det}(\eta-\Phi)= \eta^q(\eta^{2p}+a_{2}\eta^{2p-2}+\ldots+a_{2p-2}\eta^2+a_{2p}),\]
with natural involution $\sigma:\eta\mapsto-\eta$ as in \eqref{invopq}.
 \begin{itemize}
\item[i.] For $q=0$, consider  line bundles $U_{1}$ and $U_{2}$ on \[\bar \pi:\overline{S}=S/\sigma \rightarrow \Sigma,\] such that $\bar \pi_*U_1=V$ and  $\bar \pi_*U_1=W$, and express $\Lambda^{p}V$ and $\Lambda^{p}W$ and  $\Lambda^{p}V\cong \Lambda^{p}W^{*}$ in terms of $U_i$. 

%

\item[ii.] For $p=2$ and $q=0$ the quotient curve $\bar S$ of \textbf{Problem 9.i.} is a double cover of the Riemann surface $\Sigma$. Describe any additional geometric structure appearing from the fact that the Higgs bundle has low rank. 
\item[iii.] \textbf{((*))} How can the structure of an $SO(2,4)$-Higgs bundle appear through the data described in \textbf{Problem 9.ii.}?
\item[iv.] \textbf{((*))} Consider $p=2$ and $q=2$ and give a description of spectral data that could be associated to $U(2,4)$-Higgs bundles -- see \cite{cayley} for details on spectral data for orthogonal Higgs bundles with signature. 
\end{itemize}

\end{problem}
\begin{problem} \textbf{Parabolic Hitchin fibration} As explained in \cite{davidp}, the moduli space of  parahoric Higgs bundles, and in particular of parabolic Higgs bundles as defined in the previous chapter, carry a Hitchin map which is a Poisson map whose generic fibres are abelian varieties. Using the paper \cite{davidp}, think of parabolic Higgs bundles in terms of parahoric Higgs bundles:
\begin{itemize}
\item[i.] Explain the relation between parahoric Higgs bundles and parabolic Higgs bundles; 
\item[ii.]  By considering the parahoric global nilpotent cone, describe the fibre over $0$ for parabolic $\mathcal{G}$-Higgs bundles in terms the cotangent bundle of $\textrm{Bun}_{\mathcal{G}}$.
\item[iii.] Give a summery of how to establish that the Hitchin map on the moduli of polystable parabolic Higgs bundles is proper.
\item[iv.] Consider parabolic $SL(4,\C)$-Higgs bundles, and give a comparison of the Hitchin maps defined in \cite{marinap} and \cite{davidp}.

\end{itemize}

\end{problem}
  

\section{Branes in the moduli space of Higgs bundles}\label{cuarto}\label{chap-branes}

\epigraph{\textit{If people do not believe that mathematics is simple, it is only because they do not realize how complicated life is.}}{John von Neumann}


   The appearance of Higgs bundles (and flat connections) within string theory and the geometric Langlands program has led researchers to study the \textit{derived category of coherent sheaves} and the \textit{Fukaya category} of these moduli spaces. Therefore, it has become fundamental to understand Lagrangian submanifolds of the moduli space of Higgs bundles supporting holomorphic sheaves ($A$-branes), and their dual objects ($B$-branes).    
       The space of solutions to Hitchin's equations \eqref{2.1}--\eqref{hit2} is a  hyperk\"ahler manifold, and thus there is a family of complex structures from which we shall fix  $I,J,K$ obeying    quaternionic equations, following the notation of \cite{real,N1,Kap}. In particular, under this convention  the smooth locus of
 $\mathcal{M}_{G_{\mathbb{C}}}$ corresponds to the space of solutions to Hitchin's equations  with structure $I$. 
Throughout this notes we shall adopt the physicists' language in which a Lagrangian submanifold supporting a flat connection is called an {\em A-brane}, and a complex submanifold supporting a complex sheaf is a {\em B-brane}. By considering the support of branes, we shall refer to a submanifold of a  hyperk\"ahler manifold as being of type $A$ or $B$ with respect to each of the   structures, and hence one may consider branes of type 
\begin{equation}(B,B,B), ~(B,A,A), ~(A,B,A) ~\textrm{and}~ (A,A,B).\end{equation} 

We shall dedicate this chapter to the  construction, description,  and study of such branes, by considering additional structure on the Riemann surface $\Sigma$ and on the complex Lie group $G_\C$:
\begin{itemize}
\item \textbf{Branes through finite group actions on Riemann surface $\Sigma$}.  Finite group actions \linebreak $\Gamma\times\Sigma\to\Sigma$ 
 on   compact connected Riemann surfaces  have long been considered, and through these actions we shall consider $(B,B,B)$-branes appearing through $\Gamma$-equivariant representations following \cite{cmc,cmc1}.
\item \textbf{Branes through anti-holomorphic involutions}. We have seen before that anti-holomorphic involutions can be used to define real Higgs bundles, as initiated in \cite{N5}. By considering real structures on the Riemann surface and their compositions with group involutions, we shall define $(B,A,A),$ $(A,B,A)$ and $(A,A,B)$-branes following \cite{branes,real} (see also \cite{obw}). These constructions will play a special role later on when considering hyperpolygons following \cite{poli} in the final chapter of these notes.
\end{itemize}

\subsection{Branes through finite group actions.}

As shown in \cite{cmc}, a natural way to define $(B,B,B)$-branes in the moduli spaces of complex Higgs bundles is though the action of a finite group $\Gamma$ on the base Riemann surface of genus $g$.   
When the genus is 2 or 3, a complete classification of all finite group actions appears in \cite[Tables 4, 5]{B91}, 
 which allows one to perform explicit calculations in those cases. 
 Moreover, in the case of actions induced on rank 2 bundles through automorphisms of $\Sigma$, a very concrete description of the fixed points in terms of parabolic structures is given in \cite{AG}.

Within this setting, a much studied question is the appearance of   flat $\Gamma$-equi-variant $G_\mathbb{C}$-connections on $\Sigma$, for which one needs to   fix a $C^\infty$ trivialization $\underline{\mathbb{C}}^{n}=\Sigma\times\mathbb{C}^{n}\to\Sigma$ of the underlying vector bundle.  
%
%
There are a few different perspectives on  $\Gamma$-equivariant flat connection on a Riemann surface $\Sigma$, and here we shall consider the definition of \cite{cmc}, through which a $\Gamma$-equivariant flat connection is a flat connection $\nabla$ on $\underline{\mathbb{C}}^n\to\Sigma$
 such that for every $\phi\in\Gamma$ there exist a $\textrm{GL}(n,\mathbb{C})$-gauge transformation $g_\phi\colon\Sigma\to \textrm{GL}(n,\mathbb{C})$ for which 
 \begin{itemize}
 \item
 $\phi^*\nabla=\nabla \cdot g_\phi,$
 \item and  
 $\phi\mapsto g_\phi$
  a {\em generalized group homomorphism}, i.e.,  such that 
\begin{align}g_{id}&=id\\g_{(\phi\circ\tau)}(p)&=g_{\tau(p)}\circ g_{\phi}(\tau(p)).\end{align}
\end{itemize}

%

Consider $G_\C=\textrm{GL}(n,\C),\textrm{SL}(n,\C)$, and $\Gamma\times\Sigma\to\Sigma$   a finite group action by holomorphic automorphisms without fix points.
 Then,       any $\Gamma$-equivariant flat $G_\C$-connection is given by the pull-back of a flat $G_\C(n,\C)$-connection on $\Sigma/\Gamma$ (e.g. see \cite[Proposition 1]{cmc}). When the action has fixed points,    the image $B\subset \Sigma/\Gamma$ of the fixed points gives the  branch points of the ramified cover
 \begin{equation}\pi_{\Gamma} : \Sigma \rightarrow \Sigma/\Gamma, \end{equation}
 and thus  any $\Gamma$-equivariant flat $G_\C$-connection is given by the pull-back of a flat $G_\C$-connection on $\Sigma/\Gamma-B$  (e.g. see \cite[Proposition 2]{cmc}). From the definition of $\Gamma$-equivariant flat connections it follows that      a $\Gamma$-equivariant flat irreducible $G_\C$-connection  $\nabla$ on $\Sigma$, gives rise to a $\Gamma$-equivariant connection $\nabla\cdot g$ for $g\colon \Sigma\to G_\C$  a gauge transformation, which corresponds to the same point in the moduli space of flat (possibly singular) $G_\C$-connections on $\Sigma/\Gamma$ (e.g. see \cite[Lemma 1]{cmc}).

 \begin{example}
 To illustrate the appearance of equivariant connections (and thus of equivariant Higgs bundles) consider the following example from \cite{cmc}. Let $\Sigma$ be a compact Riemann surface of genus $3$   admitting a fixed point free involution $\phi\colon\Sigma\to\Sigma$. Through this action, one has    a double cover to the quotient  Riemann surface  $M:=\Sigma/\mathbb{Z}_2$ of genus $2.$
 
Let $\nabla$ be a flat unitary line bundle connection which is not self-dual,  and
such that $\phi^*\nabla=\nabla^*$. 
Then, \[\nabla\oplus\nabla^*\] is a $\mathbb{Z}_2$-equivariant flat  connection on $\Sigma$,
with a corresponding irreducible flat connection on $\Sigma/\mathbb{Z}_2$.
\end{example}

By considering different finite group actions of a group $\Gamma$ on a Riemann surface $\Sigma$  of genus  $g\geq2$, in  \cite[Theorem 8]{cmc} it is shown that the connected components of the space of gauge equivalence classes of irreducible $\Gamma$-equivariant flat $G_\C$-connections are hyperk\"ahler submanifolds
of the moduli space of flat irreducible $G_\C$-connections on $\Sigma$, and hence give $(B,B,B)$-branes in the moduli space  of $G_{\C}$-Higgs bundles. 

When $G_\C=SL(2,\C)$,  from the above, any component of the moduli space of flat $\Gamma$-equivariant $\SL(2,\C)$-connections on $\Sigma$ can be identified with the moduli space of flat $\SL(2,\C)$-connections or  $\textrm{PSL}(2,\C)$-connections on a $n$-punctured compact Riemann surface of genus $\gamma$, for $n\in \mathbb{N}$ satisfying some compatibility conditions. furthermore, more can be said about the type of fixed point sets of the induced involution on the moduli spaces of Higgs bundles by the finite group:

\begin{theorem}[\cite{cmc} Theorem 11]  Let $\Sigma$ be a compact Riemann surface of genus $g\geq2$ and $\Gamma$ be a finite group   acting on  $\Sigma$ by holomorphic automorphisms
 such that  a component of the moduli space of $\Gamma$-equivariant flat $\SL(2,\C)$-connections on $\Sigma$ has half the dimension
of the moduli space $\mathcal M_g$.
Then one  of the following holds:
\begin{itemize}
\item[(I)]  \textit{$\Gamma=\mathbb{Z}_2$ acts by a fix-point free involution on $\Sigma$, or}
\item[(II)]  \textit{$\Sigma$ is hyperelliptic of genus $3$ and $\Gamma=\mathbb{Z}_2\times \mathbb{Z}_2.$}
\end{itemize}
In the later case (II), one of the $\mathbb{Z}_2$-factors corresponds to the hyperelliptic involution, whilst the other $\mathbb{Z}_2$-factor corresponds to an involution with 4 fixed points.\end{theorem}
  
  It is interesting to note that  all genus 3 Riemann surfaces with
fixed point free actions must be hyperelliptic (see \cite[Proposition 21]{cmc}) and thus    a hyperelliptic Riemann surface of genus 3 with an additional involution with 4 fixed points is the same as Riemann
surface of genus 3 with a fixed point free involution. 

From the above, fixed point free involutions become of particular interest. In particular, for $\SL(2,\C)$-Higgs bundles, it is shown in \cite[Proposition 23]{cmc} that given a fixed point free involution $\tau$, a flat and irreducible $\tau$-invariant $\SL(2, \C)$-connection on a Riemann surface of genus 3 is
equivariant with respect to the hyperelliptic involution. It should be noted that irreducibility is not a necessary condition to be in an equivariant brane, since there exist flat reducible connections on Riemann surfaces which correspond to irreducible connections on the hyperelliptic genus 2 quotient Riemann surface surface.

The geometry of these branes can be studied, in particular, by considering their intersection with the Hitchin fibration -- as was done in the last chapter when considering real Higgs bundles. For the branes within the moduli space of $\SL(2,\C)$-Higgs bundles, one can see them expressed in terms of Prym varieties:

\begin{theorem}[\cite{cmc} Theorem 14]
Let $\tau$ be a fix point free involution. The $\tau$-equivariant $(B,B,B)$-brane intersects a generic fibre  of the corresponding Hitchin fibration over a point defining the spectral curve $S$
 in the abelian variety  $ \textrm{Prym}(S/\tau,\Sigma/\tau)/\mathbb{Z}_2.$
\end{theorem}

Following \cite{cmc} and \cite{cmc1}, one may define \textit{anti-equivariant Higgs bundles $(\bar\partial,\Phi)$} with respect to a fix point free involution $\tau$, as pairs    given by a $\tau$-equivariant holomorphic structure $\bar\partial$ 
and a Higgs field $\Phi$
which satisfies \[\tau^*\Phi=-g^{-1}\circ \Phi\circ g,\]
where $g$ is the gauge transformation such that $\tau^*\bar\partial=\bar\partial.g$
and $(g\circ\tau) g=\textrm{id}.$  These Higgs bundles form examples of Lagrangians subspaces in the moduli space of Higgs bundles, and whilst they can be defined for compact Riemann surfaces of any genus, we shall restrict our attention here to odd genus following \cite{cmc} (the more general setting is treated in \cite{cmc1}).
Indeed, it is shown in \cite[Proposition 15]{cmc} that anti-equivariant $\SL(2,\C)$-Higgs bundles for connected components in $\mathcal{M}_{\SL(2,\C)}$ which are complex submanifolds  with respect to the complex structure  $I$ and Lagrangian with respect to the complex structures $J$ and $K$, and thus define $(B,A,A)$-branes. Finally, we should mention that interesting research directions arise from \cite{cmc} when considering generalizations of the results appearing there for more general groups.

\subsection{Branes through anti-holomorphic involutions.}
Branes of different types can be obtained by considering anti-holomorphic involutions both on the compact Riemann surface as well as on the complex Lie group, and this is the perspective of \cite{branes,real} and which we shall describe in what remains of the chapter. 
Following \cite{N5} and recalling the previous analysis in the first chapter about   Higgs bundles for real groups, one has the following description of $G$-Higgs bundles considered in \cite{thesis}:
\begin{proposition}\label{invrelation}
Let $G$ be a real form of a complex semi-simple Lie group $ G_\C$, whose real structure is $\sigma$. Then,  $G${\em -Higgs bundles} are given by the fixed points in $\mathcal{M}_{ G_\C}$ of the involution $\Theta_G$ acting by
\[\Theta_G: ~(P,\Phi)\mapsto (\theta(P),-\theta(\Phi)),\]
where $\theta=\rho\sigma$, for $\rho$ the compact real form of $ G_\C$.     
\end{proposition} 
 
  One should note that  
a fixed point of $\Theta_{G}$ in $\mathcal{M}_{G_\C}$ in Proposition \ref{invrelation} gives a representation of $\pi_1(\Sigma)$ into the real form $G$ up to the equivalence of conjugation by the normalizer of $G$ in $G_\C$. This may be bigger than the Lie group $G$ itself, and thus two distinct classes in $\mathcal{M}_{G}$ could be isomorphic in $\mathcal{M}_{G_\C}$ via a complex map. 
Hence, although there is a map from $\mathcal{M}_{G}$ to the fixed point subvarieties in $\mathcal{M}_{G_\C}$, this might not be an embedding.    
A description of this phenomena in the case of rank 2 Higgs bundles is given in \cite[Section 10]{mono}, where one can see how the $SL(2,\mathbb{R})$-Higgs bundles which have different topological invariants lie in the same connected component as $SL(2,\mathbb{C})$-Higgs bundles. 

The point of view of Proposition \ref{invrelation}, which is considered throughout \cite{thesis}, fits into a more global picture described in \cite{branes, real},  through which families of $(B,A,A)$, $(A,B,A)$ and $(A,A,B)$ branes in $\mathcal{M}_{ G_\C}$ appear as the fixed point sets involutions arising from a real form $\sigma$ on a complex Lie group and a real structure $f$ on a Riemann surface.  
 A real structure\footnote{Anti-involutions on $\Sigma$ are also found in the literature as an anti-conformal maps  (e.g. see \cite{acc2})} on a compact connected Riemann surface $\Sigma$ is an 
 anti-holomorphic involution $f : \Sigma \to \Sigma$.  
The classification of real structures on compact Riemann surfaces  is a classical result of Klein, who showed that all such involutions on $\Sigma$ may be characterised by two integer invariants $(n,a)$: the fixed point set of $f$ is known to be a disjoint union of  $n$ copies of the circle embedded in $\Sigma$; the complement of the fixed point set has one or two components, corresponding to  $a=1$ or $a=0$ respectively. 
A real structure $f$ on the Riemann surface $\Sigma$ induces  involutions on the moduli space of representations $\pi_1(\Sigma) \to G_\C$, of flat connections and  of $G_\C$-Higgs bundles on $\Sigma$.

\begin{itemize} 
\item  From the Cartan involution $\theta$ of  $G$ one obtains a $(B,A,A)$-brane fixed by  
\begin{equation}
i_1(\bar \partial_A, \Phi)=(\theta(\bar \partial_A),-\theta( \Phi)).\label{i11}
\end{equation}
\item  From a real structure  $f : \Sigma \to \Sigma$   one obtains an $(A,B,A)$-brane fixed by  
\begin{equation}
i_2(\bar \partial_A, \Phi)=(f^*(  \partial_A),f^*( \Phi^*  ))= (f^*(\rho(\bar \partial_A)), -f^*( \rho(\Phi) )).
\label{i2}
\end{equation}
\item  Lastly, by looking at $i_3 = i_1 \circ i_2$, one obtains an $(A,A,B)$-brane fixed by
\begin{equation}
i_3(\bar \partial_A, \Phi)=(f^* \sigma(\bar \partial_A),f^*\sigma( \Phi)).
\end{equation}
\end{itemize}

 In what follows we shall give an overview of the above branes in terms of their geometry and topology, and describe how they can be studied though the Hitchin fibration. Generalizations of the above methods from \cite{branes,real} have now been obtained for many other moduli spaces -- e.g. see \cite{emilio1,emilio2,vicky,poli} among others.

\subsubsection{The $(B,A,A)$-brane of $G$-Higgs bundles.}\label{secinvo}

As mentioned in the first Lecture, the moduli spaces $\mathcal{M}_{ G_\C}$ have a natural symplectic structure, which we  denoted by $\omega$. Moreover, following \cite{N2}, an involution $\Theta_{G}$ from Proposition \ref{invrelation}, which corresponds to $i_1$ in \eqref{i11}, 
 sends $\omega\mapsto -\omega$. Thus, at a smooth point, the fixed point set must be Lagrangian and thus the expected dimension of $\mathcal{M}_{G}$ is half the dimension of $\mathcal{M}_{ G_\C}$. These branes of real $G$-Higgs bundles can be seen  sitting inside $\mathcal{M}_{ G_\C}$  as fixed points of $\Theta_{G}$ in the preserved fibres over   the Hitchin base $\mathcal{A}_{ G_\C}$.

By considering Cartan's classification of classical Lie algebras, we can describe the different involutions $i_1$ that can be considered on the moduli spaces of $G_\C$-Higgs bundles. For this,   we denote by $I_{p,q},~ J_{n}$  and $K_{p,q}$ the matrices
 \begin{equation}I_{p,q}=\left(
\begin{array}
 {cc}
-I_{p}&0\\
0&I_{q}
\end{array}\right)
,~~
J_{n}=\left(
\begin{array}
 {cc}
0&I_{n}\\
-I_{n}&0
\end{array}
\right)
,~~  
K_{p,q}=\left(
\begin{array}
 {cccc}
-I_{p}&0&0&0\\
0&I_{q}&0&0\\
0&0&-I_{p}&0\\
0&0&0&I_{q}
\end{array}
\right).\label{MatricesIJK}\end{equation}
  
 Recalling that the cartan involution $\Theta_G$ of Proposition \ref{invrelation} is obtained  $\theta=\rho\sigma$, for $\rho$ the compact real form of $ G_\C$ and $\sigma$ an anti-holomorphic involution of  $G_\C$ fixing a real form $G$,  we consider the following compact and split real forms: 
  \begin{table}[H]
\begin{center}
\begin{tabular}{c|c|c|c|c|c}
  $\mathfrak{g}_{c}$   $ G_\C$ & Split form &Compact  form $\mathfrak{u}$& anti-involution $\rho$ fixing $\mathfrak{u}$ & dim $\mathfrak{u}$  \\
\hline
$\mathfrak{a}_{n}$  &$\mathfrak{sl}(n,\mathbb{R})$	&$\mathfrak{su}(n)$  &$\rho(X)=-\overline{X}^{t}$& $n(n+1)$    \\
$\mathfrak{b}_{n}$  &$\mathfrak{so}(n,n+1)$		&$\mathfrak{so}(2n+1)$ &$\rho(X)=\overline{X}$& $n(2n+1)$   \\
$\mathfrak{c}_{n}$   	&$\mathfrak{sp}(2n,\mathbb{R})$	&$\mathfrak{sp}(n)$    & $\rho(X)=J_{n}\overline{X}J_{n}^{-1}$&$n(2n+1)$   \\
$\mathfrak{d}_{n}$   	&$\mathfrak{so}(n,n)$	&$\mathfrak{so}(2n)$ &  $\rho(X)=\overline{X}$& $n(2n-1)$   
\label{compact table}
 \end{tabular}
\caption{Compact forms $\mathfrak{u}$ of classical Lie algebras}
\end{center}
\end{table}

Then,  the following table allows one to see the different possible values that $\theta$ may take in terms of the anti-holomorphic involutions on the complex Lie algebras.
 \begin{table}[H]
\begin{center}
 \begin{tabular}{c|c|c|c|c}
$\mathfrak{g}_{c}$  &  Real form  $G$&  $\sigma$ fixing $G$ & $\theta=\rho\sigma$  \\
\hline
$\mathfrak{a}_{n}$ 	&$SL(n,\mathbb{R})$ &$\sigma(X)=\overline{X}$ &$\theta(X) =  -X^{t} $ \\
  &$SU^{*}(2m)$&$\sigma(X)=J_{m}\overline{X}J_{m}^{-1}$&$\theta(X) = -J_{m}X^{t}J_{m}^{-1}$\\
  &$SU(p,q)$& $\sigma(X)=-I_{p,q}\overline{X}^{t}I_{p,q}$ &$\theta(X) = I_{p,q}X I_{p,q}$\\ \hline
  $\mathfrak{b}_{n}$     &$SO(p,q)$ &$ \sigma(X)=I_{p,q}\overline{X}I_{p,q}.$&$ \theta(X) = I_{p,q}X I_{p,q} $\\\hline
 $\mathfrak{c}_{n}$ 	      &$Sp(2n,\mathbb{R})$&$\sigma(X)=\overline{X}$& $\theta(X) = J_{n}X J_{n}^{-1} $\\
  & $Sp(2p,2q)$ & $\sigma(X)=-K_{p,q}X^{*}K_{p,q}.$& $\theta(X) =K_{p,q} X K_{p,q}$\\
\hline
$\mathfrak{d}_{n}$ 	&$SO(p,q)$ &$ \sigma(X)=I_{p,q}\overline{X}I_{p,q}.$&$ \theta(X) = I_{p,q}X I_{p,q} $\\
 &$SO^{*}(2m)$  &$
 \sigma(X)=J_{m}\overline{X}J_{m}^{-1}.
$ &  $\theta(X) = J_{m}X J_{m}^{-1}$
\label{compact table2}
 \end{tabular}
\caption{Non-compact forms $G$ of classical Lie algebras $ G_\C$}\label{table1}\label{invotable}
\end{center}
\end{table}

\begin{example} One can see   $SO(p,p+q)$-Higgs bundles as fixed points of the involution  $\Theta_{SO(p,p+q)}:(E,\Phi)\mapsto (E, -\Phi)$
on the moduli space of $SO(2p+q,\mathbb{C})$ corresponding to vector bundles $E$ which have an automorphism $f$ conjugate to $I_{p,p+q}$ sending $\Phi$ to $-\Phi$ and whose $\pm 1$ eigenspaces have dimensions $p$ and $p+q$.
%
\end{example}

In what follows we shall give a summary of different appearances of the $(B,A,A)$-brane of $G$-Higgs bundles as the fixed point set  of the induced action of $i_1$ from \eqref{i11} on the Hitchin fibration. In particular, these branes may intersect the regular fibres of the Hitchin fibration in very different ways:
\begin{itemize}
\item{ \bf Finite intersection with the regular fibres}. As mentioned before, from \cite[Theorem 4.12]{thesis}, the $(B,A,A)$-brane of $G$-Higgs bundles for $G$ a split real form of $G_\C$ intersects the fibres of the Hitchin fibration in points of order two. Since we have given an overview of this case in previous chapters, we shall not consider them any further here (see \cite{thesis,mono,ortlau,cayley,isog,mas10} for papers treating the spectral data for split real forms). 

\item { \bf Positive dimensional intersection with the regular fibres}. The $(B,A,A)$-brane of $\textrm{U}(p,p)$-Higgs bundles in $\mathcal{M}_{\textrm{GL}(2p,\C)}$, as seen in Exercise \textbf{II.(4)}, and of   $\textrm{SU}(p,p)$ in $\mathcal{M}_{\SL(2p,\C)}$ intersects the regular fibres in finite dimensional abelian varieties given by Jacobians or Prym varieties of quotient spectral curves, as shown in  \cite{thesis,umm}. Having mentioned the spectral data for these objects in previous chapters, we shall dedicate the reminder of this section to the third possibility described below.
\item { \bf Empty intersection with the regular fibres}. From the structure of real Higgs bundles, as shown in \cite{thesis}, the characteristic polynomial for non-split Higgs bundles for  real forms other than  $\textrm{U}(p,p)$ and $\textrm{SU}(p,p)$ will define reducible spectral curves and thus their $(B,A,A)$-brane will lie completely over the discriminant locus of the Hitchin fibration. These branes can be shown to intersect the non-regular fibres of the Hitchin fibration in different ways, and below we shall consider three different cases:
\begin{itemize}
\item\textbf{Abelian spectral data.} Examples of these come from quasi-split but not split real forms. For these groups, it is shown that the Cameral data for the groups $\textrm{SU}(p, p + 1)$ and $\textrm{SO}(p, p + 2)$ should be abelian in \cite{peon}, and the abelian spectral data for $\textrm{SO}(p, p + 2)$-Higgs bundles is described in \cite{cayley}.

\item\textbf{Non-Abelian spectral data}. Higgs bundles for certain groups will carry natural non-abelian spectral data describing the intersection of their $(B,A,A)$-brane with the non-regular fibres of the Hitchin fibration. 
To study this situation, one can consider Higgs bundles $(E,\Phi)$ whose generic eigenspace is $r$-dimensional, and $\textrm{det}(xI-\Phi)=p(x)^r$ for some polynomial $p(x)$:  since the Hitchin map is surjective, these cases certainly occur. For $r=2$, by extending the approach from \cite{thesis},  it is shown in \cite{nonabelian}  a \textit{``nonabelianization''} procedure appearing when considering Higgs bundles   corresponding  to flat connections on $\Sigma$ with holonomy in the real Lie groups $G=\textrm{Sp}(2p,2p), ~\SL(p,\mathbb{H})$ and $\textrm{SO}(2p,\mathbb{H})$ (the latter two also known as $\textrm{SU}^*(2p)$ and $\textrm{SO}^*(2p)$ respectively), spectral data is given by certain subspace of rank 2 Higgs bundles on a spectral curve. 

\item\textbf{Abelian and Non-Abelian spectral data}. Interestingly, the generic intersection of certain $(B,A,A)$-branes with the Hitchin fibration needs abelian and non-abelian data to be described. This is the case, in particular, of $\textrm{SO}(p+q,p)$-Higgs bundles and $\textrm{Sp}(2p + 2q, 2p)$-Higgs bundles described in \cite{cayley}. In this setting, the Higgs field has generically a kernel of dimension $2q$, and its characteristic polynomial has a degree $4p$ polynomial which  generically defines a smooth curve. By considering this smooth curve and abelian data on it, as well as an auxiliary spectral cover with non-abelian data on it, a geometric description   of the brane through  \textit{``Cayley and Langlands type correspondences''}   is obtained in \cite{cayley} and will be described below.  
\end{itemize}

\end{itemize}

 In what follows we shall give some examples of branes which lie completely over the discriminant locus of the Hitchin fibration, mostly following the work in   \cite{nonabelian} and \cite{cayley}. 
 
%
%
%
%
%
%
%

\subsubsection{Nonabelianization.} Following the Definition \ref{real} consider the following:
\begin{itemize}
\item   $\SL(n,\mathbb{H})$-Higgs bundles  consisting  of a rank $2n$ symplectic vector bundle $(V,\omega)$ and a  Higgs field satisfying $\Phi=\Phi^T$ for the symplectic transpose. Using $\omega$  to identify $V$ and $V^*$, this means that $\Phi=\omega^{-1}\phi$ for $\phi\in H^0(\Sigma,\Lambda^2V\otimes K)$. 
\item  $\textrm{SO}(2p,\mathbb{H})$-Higgs bundle whose vector bundles $V=W\oplus W^*$ for a rank $2p$ vector bundle $W$. The inner product is defined by the natural pairing between $W$ and $W^*$. The Higgs field is of the form
$\Phi(w,\xi)=(\beta(\xi),\gamma(w))$ where $\beta:W^*\rightarrow W\otimes K$ and $\gamma:W\rightarrow W^*\otimes K$ are skew-symmetric. 
\item $Sp(2p,2p)$-Higgs bundles, which have vector bundle $W_1\oplus W_2$ for symplectic rank $2p$ vector bundles $(W_1,\omega_1),(W_2,\omega_2)$ and the Higgs field is of the form
$\Phi(v,w)=(\beta(w),\gamma(v))$ where $\beta:W_2\rightarrow W_1\otimes K$ and $\gamma:W_1\rightarrow W_2\otimes K$ with $\beta=-\gamma^T$, using the symplectic transpose. 
 \end{itemize}
 The spectral data for the above Higgs bundles is described in \cite{nonabelian} by means of rank 2 vector bundles, and can be summarised as follows by considering rank an $\SL(2n,\C)$-Higgs bundles $(V,\Phi)$ on a compact Riemann surface $\Sigma$ of genus $g\geq 2$, whose Higgs field has characteristic polynomial 
\begin{equation}\det(\Phi-\eta)=( \eta^n+a_2x^{n-2}+\dots+a_n)^2. \label{poly11}\end{equation}
It should be noted that Higgs bundles for the groups $\SL(n,\mathbb{H})$,  $\textrm{SO}(2p,\mathbb{H})$ and $Sp(2p,2p)$ lie within tis setting, by letting $n=2p$. By considering  $a_i\in H^{0}(\Sigma,K^{i})$ and $\eta$ the tautological section as before, the above polynomial defines   a generically smooth spectral curve 
\begin{equation}
S=\{\sqrt{\det(\Phi-\eta)}=0\}.\label{poly22}
\end{equation}

 \begin{theorem}[\cite{nonabelian}] 
Consider a rank an $SL(2n,\C)$-Higgs bundle $(E,\Phi)$ with characteristic polynomial as in \eqref{poly11}. Given  a rank $2$ vector bundle $E$ on the curve $S$ defined as in \eqref{poly22},  the  direct image of $\eta:E\rightarrow E\otimes \pi^*K$    defines a semi-stable $SL(n,\mathbb{H})$-Higgs bundle on $\Sigma$   if and only if  
  \begin{itemize}
  \item[\textbf{(i)}] 
  $\Lambda^2E\cong \pi^*K^{n-1}$ 
  \item[\textbf{(ii)}]
  $E$ is semi-stable.
\end{itemize}
When $n=2p$, the curve $S$ in \eqref{poly22} has a natural involution    $\sigma: \eta \mapsto -\eta$. In this case,  the  direct image of $\eta:E\rightarrow E\otimes \pi^*K$    defines a semi-stable $SO(2p,\mathbb{H})$-Higgs bundle on $\Sigma$ if and only if {\em \textbf{(i)}} and {\em \textbf{(ii)}} hold, as well as 
   \begin{itemize}
  \item[\textbf{(iii)}]
  $\sigma^*E\cong E$ where the induced action on  $\Lambda^2E=\pi^*K^{4p-1}$ is trivial.
  \end{itemize}
Finally,  the  direct image of $\eta:E\rightarrow E\otimes \pi^*K$    defines a semi-stable $Sp(2p,2p)$-Higgs bundle if and only if {\em \textbf{(i)}} and {\em \textbf{(ii)}} hold, as well as 
\begin{itemize}
  \item[{\bf(iv)}] $\sigma^*E\cong E$ where the induced action on  $\Lambda^2E=\pi^*K^{2m-1}$ is $-1$.
\end{itemize}
 \end{theorem}
 
 Interestingly, as noted in \cite{thesis} and \cite{nonabelian},  when considering $Sp(2p,2p)$-Higgs bundle  one has an action of $\sigma$  by $-1$ on $\Lambda^2E$, and thus at a fixed point one has distinct $+1$ and $-1$ eigenspaces. Following \cite{AG} this defines a rank $2$ bundle on the quotient  curve $\pi':\bar S:=S/\sigma \rightarrow S$ with a parabolic structure at the fixed points defined by the flag given by the $-1$ eigenspace and the parabolic weight $1/2$, thus closely related to parabolic Higgs bundles as described in Chapter 1. 
The  direct image of $E$ on the quotient curve gives a rank 2 vector bundle which can be decomposed into line bundles $E^+\oplus E^-$ through the invariant and anti-invariant sections under the induced action of $\sigma$. rank $4p$ vector bundle can be recovered through the quotient curve by taking $\bar\pi(E^+)=W_1$ and $\bar\pi(E^+)=W_2$ for $\bar \pi$ the rank $2p$-projection $\bar \pi: \bar S\rightarrow \Sigma$.  Moreover, 
through the Lefschetz fixed point formula, the degrees of $W_1,W_2$ can be expressed in terms of those of $E^\pm$.

\subsubsection{Cayley and Langlands type correspondences}. Recall that $\textrm{SO}(p+q,p)$-Higgs bundles give a $(B,A,A)$-branes which for most values of $p,q\in \mathbb{N}$, require both abelian and non-abelian data to be described within the Hitchin fibration. From Definition \ref{real}, the characteristic polynomial of a generic  $\textrm{SO}(p+q,p)$-Higgs bundle $(V\oplus W, \Phi)$ can be written as 
\begin{equation}\det(\Phi-\eta)=\eta^q( \eta^{2p}+a_1x^{2p-2}+\dots+a_p), \label{poly33}\end{equation}
and defines a generically smooth spectral curve $\pi:S\rightarrow \Sigma$ given by
\begin{equation}S=\{0= \eta^{2p}+a_1x^{2p-2}+\dots+a_p\}, \label{poly44}\end{equation}
taking $a_i\in H^0(\Sigma,K^{2i})$ and $\eta$ the tautological section of $\pi^*K$.
As in other cases, since the polynomial defining $S$ in \eqref{poly44} only has even coefficients, there is a natural involution $\sigma:\eta\mapsto -\eta$ through which one can define the quotient curve $\bar \pi: \bar S:=S/\sigma \rightarrow \Sigma$.

From  \cite[Theorem 4.11]{cayley},  there is a one-to-one correspondence between $\textrm{SO}(p+q,p)$-Higgs bundles and the triple of spectral data $(L,M,\tau)$ given by:
 \begin{itemize}
 \item[(i)] A line line bundle $L\in \textrm{Jac}(\bar S)[2]$, space of torsion two line bundles of degree zero on $\bar S$;
 \item[(ii)] A rank $q$ equivariant orthogonal bundle $M$ of type $(q-1,1)$ on the 2-fold auxiliary   curve
 $
 C:=\{0=\eta^2-a_p\}\subset \textrm{Tot}(K^p),
$ 
 for   $a_p$ as in \eqref{poly44},   with a choice of orientation, and satisfying the following semistability condition:  all invariant isotropic subbundles $M'\subset M$ have  $\textrm{deg}(M')\le 0$;
 \item[(iii)] A choice of isometry $\tau$ over each zero $x$ of $a_1$, between the corresponding fibres of $W$ and those of the $-1$-eigenspace $M^-$ of $M$, giving the extension data needed between (i) and (ii).  This extension data may be identified with a collection $\{ \tau_x \}$ of unit vectors $\tau_x \in L_x$ for each ramification point $x$ of $C$.
\end{itemize} 
In the above description, the spectral curve $C$ has a sheet swapping involution $\sigma_C : C \to C$ defined by $\sigma_C(\eta)=-\eta$. Then, the equivariant orthogonal bundle $(M , Q_M , \tilde{\sigma}_C )$ is said to have type $(q-1,1)$ if the lift of the involution to $M$ has $\pm 1$ eigenspaces of dimensions  $(q-1,1)$. 

From the above, one can see the intersection of the $(B,A,A)$-brane of $\textrm{SO}(p+q,p)$-Higgs bundles with a fibre of the $\textrm{SO}(2p+q,\C)$ Hitchin fibration over a point $a$ is identified with a   covering of the product of two moduli spaces:
\begin{equation}\label{cay}
\mathcal{M}_{Cay}(a) \times  \mathcal{M}_{Lan}(a).
\end{equation}
The covering in question corresponds to  extension data $\tau$ in (iii) above. The Cayley and Langlands moduli spaces $\mathcal{M}_{Cay}(a)$ and $\mathcal{M}_{Lan}(a)$ are given as follows:
\begin{itemize} 
\item $\mathcal{M}_{Cay}(a)$ is a fiber of the Hitchin map for the moduli space of $K^2$-twisted $GL(p,\mathbb{R})$-Higgs bundles on $\Sigma$, which can be identified with $\textrm{Jac}(\bar S)[2]$ in (i) above.
\item $\mathcal{M}_{Lan}(a)$ is the moduli space of equivariant $SO(q)$-bundles $M$ on $C$  in (ii) above
\end{itemize}

As explained in \cite{cayley}, the reconstruction of the $\textrm{SO}(p+q,p)$-Higgs bundle $(E,\Phi)$   from the spectral data 
involves taking an extension of the form
\[
0 \to V_0 \to E \to F \otimes K^{1/2} \to 0,
\]
where $F$ is the $Sp(2p,\mathbb{R})$-Higgs bundle associated to the Cayley moduli space, and $V_0$ is the invariant direct image under $\pi_C : C \to \Sigma$ of the equivariant orthogonal bundle $M$ on $C$. The \textit{extension data $\tau$} is given by the extension class defining the above extension, leading to the following conjecture. 

\begin{conjecture}The extra components in the moduli space of $SO(p+q,q)$-Higgs bundles conjectured to contain positive representations by  Guichard and Wienhard \cite[Conjecture 5.6]{anna2} are those connected components containing Higgs bundles whose spectral data $(L , M , \tau )$  has the form $( \overline{\pi}^* I , \mathcal{O}^{\oplus q} , \tau)$, where $I$ is a non-trivial $\mathbb{Z}_2$-line bundle on the Riemann surface $\Sigma$.  
\end{conjecture}

\subsubsection{The $(A,B,A)$-brane from a real structure on the Riemann surface.}\label{secinvobis}
As mentioned before, given a real structure  $f : \Sigma \to \Sigma$  there is a natural  $(A,B,A)$-brane of Higgs bundles described in \cite{branes,real} appearing as the fixed point set of the induced involution   
\begin{equation}
i_2(\bar \partial_A, \Phi)=(f^*(  \partial_A),f^*( \Phi^*  ))= (f^*(\rho(\bar \partial_A)), -f^*( \rho(\Phi) )).
\end{equation}

These branes can be studied through the spectral data associated to $G_{\C}$-Higgs bundles introduced in \cite{N2}, and considering the $(A,B,A)$-brane as sitting inside the corresponding Hitchin fibration. By doing so,  new examples of real integrable systems are obtained in \cite{branes} over a real subspace  of the base of the  Hitchin fibration:

\begin{theorem}[\cite{branes} Theorem 17] \label{teolan}
If $(A,B,A)$-brane contains smooth points then the intersection of the brane with the Hitchin fibration  gives a Lagrangian fibration with singularities. The generic fibre is smooth and consists of a finite number of tori.
\end{theorem}

 In the case of $G_\C=\textrm{GL}(n,\mathbb{C})$, the number of components that each fibre in Theorem \ref{teolan} has can be expressed in terms of the number of components of the fixed point set of an induced involution on the spectral curves (\cite[Proposition 32]{branes}). Moreover in the rank $2$ case the calculations lead to an exact formula in terms of the associated quadratic differential:
  
 \begin{theorem}[\cite{branes}] 
Let $q \in H^0(\Sigma , K^2)$ be a quadratic differential which has only simple zeros and such that $f^*(\overline{q}) = q$, and consider the spectral curve
\[S=\{0=\eta^2 -q\}\subset \textrm{Tot}(K).\] Let $n_+$ be the number of fixed components of $f$ on which $q$ is non-vanishing and positive, and $u$ the number of zeros of $q$ which are fixed by $f$. Then the  number of components of the induced involution on the generic fibre $Jac(S)$ of the $GL(2,\C)$ Hitchin fibration, giving the components of the $(A,B,A)$-brane, is $2^d$, where 
\begin{equation}
d = \left\{
\begin{array}{ccc}2n_+ + \frac{u}{2} - 1&\textrm{if}& 2n_+ + \frac{u}{2} > 0\\
 1& & \textrm{otherwise}\end{array}.\right.\end{equation}
 \label{teo111}
\end{theorem}

 Moreover,   for $\textrm{SL}(2,\mathbb{C})$ can be specified further when considering the Prym variety $\textrm{Prym}(S,\Sigma)$, as seen in \cite[Theorem 38]{branes}.

\begin{theorem}[\cite{branes}]
Considering the setting of Theorem \ref{teo111} and suppose that at least one branch point of $S$ is fixed by $f$. Then the fixed point set of the action of $\tilde{f}$ on   $\textrm{Prym}(S,\Sigma)$ has \[2^{n_0 + \frac{u}{2} - 1}\] connected components, where $n_0$ is the number of fixed components of $f : \Sigma \to \Sigma$ which do not contain branch points, and $u$ is the number of branch points which are fixed by $f$.
\end{theorem}

The study of these $(A,B,A)$-branes allows one to consider the following problem: given a 3-manifold $M$ with boundary $\Sigma$, by restricting representations $\pi_1(M)$ to the boundary gives a natural map \[\textrm{Rep}(\pi_1(M), G_\C)\rightarrow\textrm{Rep}(\pi_1(\Sigma), G_\C),\] and one can show that representations that extend from the boundary to the three manifold always define an $(A,B,A)$-brane inside $\MGC$. 
Whilst a general description of these $(A,B,A)$-branes appearing though 3-manifolds is not available, a first  step towards understanding them is to consider the $(A,B,A)$-branes constructed in \cite{real} as part of a triple through \eqref{i2}, fixed by the involution $i_2$ in \eqref{i2}. Indeed, consider the   the 3-manifolds  \begin{equation}M= \Sigma \times [-1,1]/ (x,t)\mapsto (f(x),-t),\end{equation}  a form which any handle body can be put in,  one can see that the quotient $M$ is a $3$-manifold with boundary $\partial M = \Sigma$.  
 
\begin{theorem}[\cite{real}] \textit{Let $(E, \Phi)$ be a fixed point of $i_2$ with simple holonomy. Then the associated connection extends over $M$ as a flat connection iff the class $[E] \in \tilde{K}^0_{\mathbb{Z}_2}(\Sigma)$ in the reduced equivariant $K$-theory is trivial.}
\end{theorem}

\subsection{Problem set III.}
This problem set is meant to make the  reader be acquainted with actions on Riemann surfaces and on groups, and their induced actions on the moduli space of Higgs bundles. They provide a mixture of problems from the lectures as well as some additional ones.  

\begin{problem} textbf{Topological invariants}. The topological invariants associated to $(W_1\oplus W_2, \Phi)$ in the $(B,A,A)$-brane of $U(p,q)$-Higgs bundles are the degrees $\text{deg}(W_1)$ and   $\text{deg}(W_2)$. In the case of $p=q$, from \cite{umm} the invariants  can be seen in terms of the degree of the line bundle $L$ on $S$ and the number of ramification points of $S$ over which the linear action of $\sigma$ on the fibre of $L$ is $-1$. These invariants satisfy several properties we shall see here:
\begin{itemize}
 \item[(i)] Use the Lefschetz fixed point formula in \cite[Theorem 4.12]{at}  to see that the parity of the degree of $L$ and the number of points over which $\sigma$ acts as $-1$ needs to be the same. \label{exumm}

\item[(ii)] Following \cite{brad}, the Toledo invariant $\tau(\text{deg}(W_1),\text{deg}(W_2))$ associated to $U(p,p)$-Higgs bundles is defined as $\tau(\text{deg}(W_1),\text{deg}(W_2)):=\text{deg}(W_1)-\text{deg}(W_2)$. Use Exercise \ref{exumm} to express the invariant in terms of fixed points of $\sigma$ and obtain natural bounds.

\item[(iii)] \textbf{((*))} In the case of $SU(p,p)$-Higgs bundles with maximal Toledo  Toledo invariant    the spectral data is given by  a covering of $\text{Prym}(S/\sigma, \Sigma)$. For $SU(p,p+1)$-Higgs bundles, the Cameral data can be used to obtain the spectral data as seen in \cite{ana1}.
Use the tools of \cite{umm} together with those of \cite{xia}  to obtain a description of the $SU(p,1)$ spectral data.
\end{itemize}

\end{problem}\begin{problem}
\textbf{$(B,B,B)$-branes from finite group.} Recall that the dimension of the moduli space of flat  space is $\SL(2,\C)$-connections on a Riemann surface $\Sigma$ of genus $g$ is 
$\dim \mathcal M_g=6g-6,$ 
and the dimension of  the moduli space $\mathcal M_{\gamma,n}$ of flat $\SL(2,\C)$-connections on a $n$-punctured Riemann surface of genus $\gamma$ 
with fixed local monodromies (with simple eigenvalues) is
$
\dim \mathcal M_{\gamma,n}=6\gamma-6+2n,
$  provided that either $\gamma\geq 2$, or that $\gamma\geq1$  and $n\geq1$,  or finally that  $\gamma=0$ and $n\geq4$. Otherwise, the dimension of $\mathcal M_{\gamma,n}$  is either $0$ or $2$. 
 
 \begin{itemize}
 \item[i.]  Use Riemann-Hurwitz to give a relation between the genus $g$ and $\gamma$ of the Riemann surfaces $\Sigma$ and $\Sigma/\mathbb{Z}_p$.
 
\item[ii.] Let  $g\geq2$, and $\Gamma$ be a  finite group of order $h$ acting on  $\Sigma$ by holomorphic automorphisms. Using \textbf{Problem 12.i.}, show that if a component of the moduli space of equivariant flat $\SL(2,\C)$-connections on $\Sigma$ has half the dimension of  $\mathcal M_g$, then $h=2^k$ for some  natural number $k\in\mathbb{N}.$

  \item[iii.]  What can you say about the genus $g=1$ case in \textbf{Problem 12.ii.}?
 \end{itemize}

\end{problem}\begin{problem} \textbf{Equivariant and anti-equivariant spectral data.}
Given a compact Riemann surface $\Sigma$,  double covers of $\Sigma$ can be constructed through holomorphic sections $s\in H^0(\Sigma,L^2)$, for $L$ a holomorphic line bundle on $\Sigma$. Restricting to sections which have simple zeros, there is a  double cover 
$\pi\colon S\to \Sigma$ branched over the zeros of $s$ 
such that a square root $t\in H^0( S,\pi^*L)$ satisfying $t^2=s$ exists. 

 \begin{itemize}
 \item[i.] When considering the Hitchin fibration,   the smooth loci in the Hitchin base is given by sections in $H^0(\Sigma,K^2)$ with simple zeros. Consider instead $L=K^2$ and prove the there is a (unique) double cover $\bar \pi: \bar S\to \Sigma$ branched over the zeros of $s\in H^0(\Sigma,L^2)$ such that a square root $t\in H^0( S,\pi^*L)$ satisfying $t^2=s$ exists. 
 \item[ii.] Note that for $\Phi$ an $SL(2,\C)$-Higgs field with spectral line bundle $L$, the Higgs field $\Phi$ has spectral line bundle $\sigma^* L$, where $\sigma$ is the natural involution switching the sheets of the 2-fold spectral curve.  Show that anti-equivariant  $SL(2,\C)$-Higgs bundles with respect to a fixed point free involution $\tau$ acting on $\Sigma$ have spectral data  determined by 
 \begin{equation}\label{non-sym-prym}
\mathcal{P}^{\vee}:=\{L\in Jac(S)\mid \sigma^*L=L^*,\tau^*L=L^*\}\subset \textrm{Prym}(S,\Sigma).
\end{equation}
 \item[iii.] Recalling that Higgs bundles for split real forms correspond to torsion two points in the Hitchin fibres for complex Higgs bundles, show a similar result to \textbf{Problem 13.ii.} for $SL(2,\R)$-Higgs bundles.    
 \end{itemize}

\end{problem}\begin{problem} \textbf{$(B,A,A)$-branes from real structures.} We have seen how real Higgs bundles can be defined through involutions on the Lie algebras. These actions can also be considered to understand properties of the spectral data description of the brane, as well as of the topological invariants labelling components of the brane. 
\begin{itemize}
\item[i.] The associated involution $\theta$ from Table \ref{invotable} decomposes $\mathfrak{u}=   \mathfrak{h}\oplus i \mathfrak{m}$. Give an explicit description of $\mathfrak{m}$ and $\mathfrak{h}$ and of the real form $\mathfrak{g}=   \mathfrak{h}\oplus  \mathfrak{m}$.
\item[ii.]  \textbf{((*))}
Nonabelianization of the fixed point set of $\Theta_{Sp(2p,2p)}$ can also be seen through Cameral covers, as shown in \cite{ana1}. Realise the action of $\sigma$ in terms of Cameral covers. 

\item[iii.]  \textbf{((*))} Considering the notion of ``strong real form'' from \cite{adam}, describe the corresponding Higgs bundles and give a definition of $\Theta_{G}$ for which one does not have the problem described in the above paragraph. 
\end{itemize}
\end{problem}\begin{problem}\textbf{$(A,B,A)$-branes from real structures.} Let  $\Sigma$ be of genus 2 with a real structure $f$ whose  invariants are  $(n,a)=(3,0)$.
%

 Consider an $SL(2,\C)$-Higgs bundle with  spectral curve $S =\{0=\eta^2 - q\}$ for $q$ a quadratic differential with simple zeros. For $z\in \Sigma$, the involution $f$ induces an involution 
$\tilde{f}(\eta , z) = (f^*(\overline{\eta}) , f(z) ).$

\begin{itemize}
\item[i.] What is the effect of replacing $q$ by $\alpha q$ for $\alpha\in \mathbb{R}_+$? What about replacing $q$ by $-q$?

\item[ii.] Suppose $\Sigma$ is hyperelliptic and describe the system of quadratic differentials on $\Sigma$ through two points $a_1,a_2 \in \mathbb{P}^1$.
\item[iii.] What are the conditions on $a_{1}$ and $a_{2}$ from \textbf{Problem 15.ii.} to obtain a quadratic differential $q$ with $f^*(\overline{q}) = q$?
\end{itemize}

\end{problem}

\section{Higgs bundles and correspondences}\label{polii}
 
\epigraph{\textit{In recent times, a good number of the most striking resolutions of longstanding problems have come about from an unexpected combination of ideas from different areas of mathematics.}}{Sir Michael Atiyah}

When considering the moduli space of Higgs bundles for complex groups, one may ask which correspondences it has with other moduli spaces (of Higgs bundles or of other mathematical objects), and whether any of its subspaces have interesting correspondences and interpretations. In what fallows we shall describe a few of these for the whole moduli spaces of Higgs bundles, as well as for  branes of Higgs bundles as defined  in Chapter \ref{cuarto}. 

\subsection{Group homeomorphisms.}
Given a fixed Riemann surface $\Sigma$ and a homomorphism between two Lie
groups \begin{equation}\Psi : G_\C \rightarrow G_\C',\label{iso1}\end{equation} there is a naturally induced map between representation spaces
\[\Psi : Rep(\Sigma, G_\C) \rightarrow Rep(\Sigma, G_\C'),\]
for $Rep(\Sigma, G)$  the space of representations modulo conjugation. As mentioned in the beginning of the notes,  there is a correspondence between surface group representations and Higgs bundles, and hence there should be similar induced morphism of $G_\C$-Higgs bundle moduli spaces.  Therefore, one may ask which correspondences there are between Higgs bundles from group homomorphism, and which of those maps descend to correspondences between subgroups.  


Through Definition \ref{complex}, one can see that the map in \eqref{iso1} induces a map on principal  $G_\C$-Higgs bundles $(P_{G_\C},\Phi)$ given by
\begin{equation}
\Psi: (P_{G_\C},\Phi)\mapsto (P_{G_\C}\times_\Psi {G'_\C}, d\Psi(\Phi)),\label{iso2}
\end{equation}
where $P_{G'_\C}:=P_{G_\C}\times_\Psi {G'_\C}$, and the image of the Higgs field is $\Phi':=d\Psi\Phi$ for \[d\Psi: ad(P_{G_\C})\rightarrow ad(P'_{G_\C})\]  defined by the derivative at the identity of the map $\Psi$.

\subsubsection{Induced maps from isogenies.} A first step towards tackling the above question is to consider the case where the group homomorphism is given by the  isogenies of Lie groups, as we shall see next.
Through  correspondences between low rank Lie algebras, one can discover some interesting relations between different Higgs bundles. Some of the low rank isogenies between complex are  \begin{equation}\begin{array}{ccl}
\SO(3,\C)&\cong&\Sp(2,\C)~~\cong~~ \SL(2,\C);\nonumber\\
\SO(4,\C)&\cong& \SL(2,\C)\times \SL(2,\C);\nonumber\\
\SO(5,\C)&\cong& \Sp(4,\C);\nonumber\\
 \SO(6,\C) &\cong&\SL(4,\C).\nonumber
\end{array}
\end{equation}
 These correspondences induce maps between real Lie algebras, which can be described as follows: 
\begin{equation}\begin{array}{lclcccccccl}
(a)~\mathfrak{so}(2,1)&\cong&\mathfrak{sl}(2,\R), 
 &&(f)~\mathfrak{so}(3,3)&\cong&\mathfrak{sl}(4,\R),
\nonumber \\
(b)~\mathfrak{so}(2,2)&\cong&\mathfrak{sl}(2,\R)\times \mathfrak{sl}(2,\R),
&&(g)~\mathfrak{so}(2,4)&\cong&\mathfrak{su}(2,2),
\nonumber\\
(c)~~\mathfrak{so}^*(4)&\cong&\mathfrak{su}(2)\times \mathfrak{sl}(2,\R),
&&(h)~\mathfrak{so}(1,5)&\cong&\mathfrak{su}(4),
 \nonumber\\
 (d)~\mathfrak{so}(1,3)&\cong&\mathfrak{sl}(2,\C),
  &&(i)~  \mathfrak{so}(2,3)&\cong&\mathfrak{sp}(4,\R).
  \nonumber\\
 (e)~ \mathfrak{so}(1,4)&\cong&\mathfrak{sp}(2,2),
\end{array}
\end{equation}

 The rank 2 and 3 isogenies in (b) and (f) inducing maps between branes were described in \cite{isog} through the regular fibres of the Hitchin fibration,  whereas \cite{mas10} will take care of the correspondences in (c) and (d) between branes inside singular fibres. All of the other correspondences have never been studied through spectral data, and thus they provide a rich list of open questions one may ask.  In what follows, we shall give an overview of this area and the questions one may ask, based on \cite{isog} and  \cite{mas10}. For further background, the reader may want to refer to \cite{ap} where some further explanations about isogenies in terms of Higgs bundles are given.

From \eqref{iso2}, it is not hard to see that the correspondence between Higgs bundles at the level of the vector bundle and Higgs field can be deduced naturally (see \cite{ap,isog}). Recall that  an $SL(n,\C)$-Higgs bundle can be seen as is a pair $(E,\Phi)$  for  $E$ a rank $n$ holomorphic bundle on $\Sigma$ with trivial determinant, and $\Phi$   a traceless Higgs field. Moreover,  an $\SO(n,\C)$-Higgs bundle can be written as a pair $(E,\Phi)$  where $E$ is a rank $n$ holomorphic bundle on $\Sigma$ with an orthogonal structure $Q$ and a compatible trivialization of its determinant bundle, where a trivialization $\delta:\det(E)\simeq\mathcal{O}_{\Sigma}$ is compatible with $Q$ if $\delta^2$ agrees with the trivialization of $(\Lambda^nE)^2$ given by the discriminant $Q:\Lambda^nE\rightarrow\Lambda^nE^*$, and the Higgs field $\Phi$ satisfies  {$Q(u,\Phi v)=-Q(\Phi u, w)$}.


\begin{proposition}\label{mapa1}
The  isogenies between Lie groups of rank 2 and 3 given by
 \begin{align}\label{isogenies}
\mathcal{I}_2&:\SL(2,\C)\times\SL(2,\C)\rightarrow\SO(4,\C);\\
\mathcal{I}_3&:\SL(4,\C)\rightarrow\SO(6,\C);
\end{align} 
induce maps on $\SL(2,\C)$-Higgs bundles  and $\SL(4,\C)$-Higgs bundles given by
 \begin{align}
 \mathcal{I}_2&: (E_1,\Phi_1), (E_2,\Phi_2) \mapsto \left[E_1\otimes E_2, \Phi_1\otimes 1+1\otimes\Phi_2\right]\in \mathcal{M}_{SO(4,\C)} ,\\
 \mathcal{I}_3 &:(E,\Phi)\mapsto \left[\Lambda^2 E, \Phi\otimes 1+1\otimes\Phi\right]\in \mathcal{M}_{SO(6,\C)}.
\end{align} 

\end{proposition}
\begin{proof} This follows directly from \eqref{iso1} and \eqref{iso2} applied to the particular $G_\C$-Higgs bundles. 
\end{proof}
 
Our interest is in the understanding of the maps described in Proposition \ref{mapa1} through spectral data. For this, one needs to study the correspondence between eigenvalues of the Higgs field, giving the correspondence between spectral curves, and between the eigenspaces, giving the line bundles on the spectral curves.

 In order to understand the relation between eigenvalues for Higgs bundles through group homomorphisms, it is useful to recall that   the characteristic polynomial of a Higgs field $\Phi$ defines a curve in the total space of the canonical bundle $K$ whose equation can be obtained by considering
\begin{equation}p(x)=x^n + \textrm{tr} (\Phi) x^{n-1}+ \sum_{k=2}^{n-1} x^{n-k}~(-1)^k ~\textrm{Tr}(\Lambda^k \Phi) +(-1)^n~\textrm{det}(\Phi)\label{uno} \end{equation}
where $\textrm{Tr}(\Lambda^k \Phi) $ is the trace of the $k$th exterior power of $\Phi$, with dimension $\binom{n}{k}$, as in \eqref{uno1} before.
 
 As seen before, and explained in \cite{N2}, the spectral curve $\pi: S\rightarrow \Sigma$ associated to an $\SL(n,\mathbb{C})$-Higgs bundle   $(E,\Phi)$ has equation 
 \begin{equation}
S:=\{0=\text{det}(I \eta-\Phi)\}=\{ 0=\eta^{n}+a_{2}\eta^{n-2}+\ldots + a_{n-1}\eta+a_{n}\},\label{spectralsl}
\end{equation}
for $a_{i}\in H^{0}(\Sigma,K^{i})$, and  $\eta$  the tautological section of $\pi^*K$.  The spectral curve $S$ is generically smooth and the fiber over a generic point in the Hitchin base can be identified   with  the Prym variety  $\textrm{Prym}(S,\Sigma)$.
For a fixed choice of $K^{1/2}$, the vector bundle $E$ is recovered  from $L\in  \mathrm{Prym}(S,\Sigma)$ as
$E:=\pi_*(L\otimes\pi^*(K^{(n-1)/2})).
$

For $SO(2n,\C)$-Higgs bundles, the Higgs field $\Phi$ defines a  cover $\pi:S\rightarrow \Sigma$ given by \eqref{spectralsl} where $a_n=p_{n}\in H^{0}(\Sigma, K^{n})$ is the Pfaffian of $\Phi$. 
 The curve is singular  at   $\eta=p_n=0$, which are the fixed points of  a natural involution $\sigma:\eta\mapsto -\eta$. The normalization  
 $\hat{\pi}:\hat{S}\rightarrow \Sigma$ of $S$, is what we shall refer to as the spectral curve.  This curve $\hat{S}$ is generically  smooth, and  the  involution $\sigma$ extends to an involution $\hat{\sigma}$  which does not have fixed points.   The generic $\SO(2n,\C)$ fibers  by a connected component of  $\textrm{Prym}(\hat{S},\hat{S}/{\sigma})$. Similarly to the previous case, for a choice of $K^{1/2}$,  the vector bundle $E$ is recovered from $L\in \textrm{Prym}(\hat{S},\hat{S}/{\sigma})$ as 
 \begin{equation}\label{EfromL-SO}
E:=\pi_*(L\otimes (K_{\hat{S}}\otimes\pi^*(K^*))^{1/2}).
\end{equation}

Through the above spectral data, together with the correspondence given in Proposition \ref{mapa1}, the isogenies induce the following  map on the Hitchin fibrations from \cite[Theorem 1 $\&$ Theorem 2]{isog}.
Moreover, the morphisms from Theorem \ref{mapa2} restrict to real (e.g. see   Problems \ref{prob1} and \ref{prob2}). 
\begin{theorem}[\cite{isog}]
\label{mapa2}
 Consider  two generic $SL(2,\C)$-Higgs bundles $(E_i,\Phi_i)$, whose spectral data is $S_i$ and  $L_i\in Prym(S_i,\Sigma)$. Then $\mathcal{I}_2$ is a $2^{2g}$ fold-covering onto its images and the spectral data for the image   $\mathcal{I}_2((E_1,\Phi_1), (E_2,\Phi_2))$ is given by $(\hat{S}_4,\mathcal{L}_4)$ where 

\begin{itemize}
\item $\hat S_4:=S_1\times_{\Sigma}S_2$ is a smooth ramified fourfold cover, and
\item $\mathcal{L}_4:= p_1^*(L_1)\otimes p_2^*(L_2)$ where $p_i:S_1\times_{\Sigma}S_2\rightarrow S_i$ are the projection maps.
\end{itemize}
Moreover, for a generic $SL(4,\C)$-Higgs bundle $(E',\Phi')$ with spectral data $(S',L')$, the $\SO(6,\C)$-Higgs bundle given by the image $\mathcal{I}_3(E,\Phi)$ has spectral data    $(\hat{S}_6,\mathcal{I}_3(L))$, where
\begin{itemize}
\item $\hat{S}_6$ is the symmetrization of the non-diagonal component in   $S'\times_{\Sigma} S'$; 
\item $\mathcal{I}_3(L)$ is a canonical twist  generated by 
local sections of $\mathcal{L}':=p_1^*(L)\otimes p_2^*(L)$ that are anti-invariant with respect to the symmetry of $S'\times_{\Sigma} S'$.\end{itemize}
\end{theorem}
It should be noted that the Higgs bundles considered in Theorem \ref{mapa2} lie all within the generic fibres of the Hitchin fibration. However, the isogeny descends to all the real forms of those complex Lie groups, and thus one should be able to obtain the morphisms between real Higgs bundles which lie completely over the singular locus of the Hitchin fibration. The first cases studied, which involve the geometry of ribbons, appear in \cite{mas10}.
 \begin{remark}
The approach of \cite{isog} was also of use when studying problems in $F$-theory related to degenerations of elliptically fibred CY 3-folds as done  in \cite{tbrane}.  
 \end{remark}
\subsection{Polygons and Hyperpolygons.}
An interesting relation between mathematical objects appears when considering \textit{parabolic Higgs bundles} (as defined in Section \ref{secpara}) and \textit{quiver varieties} (representation varieties of quivers, for which the base manifold is a point), as studied in \cite{31,33} and subsequent work of many researchers.  

\begin{definition}
A \textrm{quiver}   $Q=(V,A,h,t)$ is an oriented graph, consisting of a finite vertex set $V$, a finite arrow set $A$, and head and tail maps $h,t:A\rightarrow V$.
\end{definition}

\begin{example}Quivers may be represented through oriented graphs, and an example that will appear useful in coming sections is that of a star shaped quiver as bellow, where the set of vertices is $V=\{v_1,\ldots,v_3\}$, and the set of arrows is $A=\{a_i:v_i\rightarrow v_0~|~1\leq i\leq3\}$:
\begin{equation}\xymatrixcolsep{0.3pc}\xymatrixrowsep{0.3pc}
\xymatrix{ & & & \: \: \: \: \bullet v_1 \ar[ddd]_{a_1}& & & 
\\ &  & & & &  & \\ \\ v_{3} \bullet \ar[rrr]_{a_{3}} & & &  \bullet v_0 & & & \bullet v_2\ar[lll]^{a_2}  }\label{exquiv}\end{equation}
Here, the head map is given by $h(a_i)=v_0$, and the tail map by $t(a_i)=v_i$. 

\label{ex1}
\end{example}
The \emph{dimension vector} of $Q$ is the $(n+1)$-tuple of positive integers $d=(r,1,\dots,1)$, and this vector determines $Q$ uniquely as a star quiver. Moreover, quivers carry naturally defined spaces of representations:\begin{equation}
  \textrm{Rep}(Q) = \bigoplus_{i=1}^n \textrm{Hom}(\C, \C^r) \cong  \textrm{Hom}(\C^n, \C^r) = \C^{n\times r}.\nonumber
\end{equation} Of particular interest for us are star-shaped quivers $Q_n$,  with one central vertex $v_0$ and $n$ outer vertices $v_0, \dots , v_n$ and arrows $a_i : v_i \rightarrow v_0$, as the one in Example \ref{ex1} which has $n=3$. In what follows we shall consider  fixed  dimension vector $d=(r,1,\cdots,1) \in \mathbb{N}^{n+1}$, and recall the value by writing  the star-shaped quiver as $Q^r_n$.  

\begin{remark} A representation of a quiver $Q$ can also be written as 
\[W:=((W_v)_{v\in V}, (\phi_a)_{a\in A})\] where $W_v$ is finite dimensional vector space; and 
   $\phi_a: W_{t(a)}\rightarrow W_{h(a)}$ is a linear map for all $a\in A$. \end{remark}

    Similar   to when working with parabolic Higgs bundles, one may define the slope and stability conditions of quivers. 
For a stability parameter $\alpha:=(\alpha_v)_{v\in V}$ in $\mathbb{Z}^V$,  the $\alpha$-slope of $W$ is 
 \begin{equation}\mu_\alpha(W) := \frac{\sum_{v\in V} \alpha_v\dim W_v}{\sum_{v\in V} \dim W_v} \in \mathbb{Q}.\end{equation}
Then, the representation $W$ is said to be:
  \textit{$\alpha$-semistable} if $\mu_\alpha(W') \leq \mu_\alpha(W)$ for all  sub-representation $0 \neq W'\subset W$;
  \textit{$\alpha$-stable} if $\mu_\alpha(W') < \mu_\alpha(W)$ for all  subrepresentations $0 \neq W' \subsetneq W$.
There are natural notions of Jordan-H\"{o}lder filtrations (as those considered in Section \ref{complex}), and  two $\alpha$-semistable  representations of $Q$ are said to be $S$-equivalent if their associated graded objects for their Jordan--H\"{o}lder filtrations are isomorphic. For more details, the reader may refer, for example, to \cite{vicky} whose notation we will follow for most of this section. 

Our interest in quiver varieties comes from their relation to vector bundles and Higgs bundles which can be obtained, in particular, considering the space of polygons and hyperpolygons (e.g., see  \cite{ginzburg} and the references therein).

 The polygon space $\mathcal{M}_{\text{poly}}(n,\ell)$ consisting of $n$-gons in the Euclidean space $\mathbb{E}^3$ with lengths given by $\ell=(\ell_1,\ldots,\ell_n)$ modulo orientation-preserving isometries is homeomorphic  the moduli space of $\ell$-semistable $n$ ordered points on $\mathbb{P}^1$  modulo the automorphisms of the projective line   (see \cite{vicky} and references therein). Moreover, it is also  isomorphic to a moduli space of representations of a star-shaped quiver $Q^\ell_n$.
 

 In order to define the moduli space of hyperpolygons, we shall consider a quiver $Q = (V,A,h,t)$ and define the doubled quiver as
$\overline{Q}:=(V,\overline{A}, h,t)$ where $\overline{A} = A \cup A^*$ for \[ A^* := \{ a^* : h(a) \rightarrow t(a)\}_{a \in A}.\] Through the double quiver of a star shaped quiver,   the authors of \cite{steve} constructed the hyperpolygon spaces for any value of $r$,
 which are the hyperk\"ahler analogues of polygon spaces (see references therein for previous work appearing in a less general manner).

{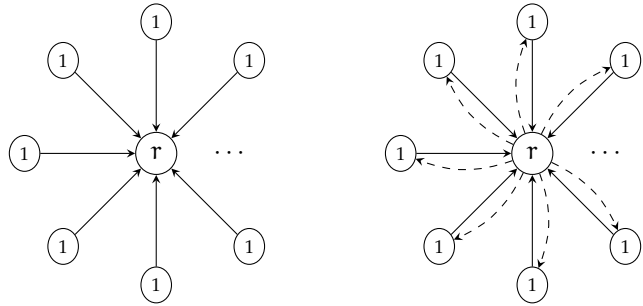
\begin{figure}[H]\centering\begin{tikzpicture}
  \node[draw, ellipse] (Lsink) at (0,0) {$r$};
  \node[draw, ellipse] (Rsink) at (5,0) {$r$};
   \foreach \a in {45}
   \foreach \b in {90 }
     \foreach \c in {135 }
       \foreach \d in {180 }
     \foreach \e in {225 }
            \foreach \f in {270 }
     \foreach \g in {315 }
 {
    \node[draw, ellipse, scale=0.7] (Lv) at ($(Lsink)+(\a:1.75)$) {$1$};
        \path[-stealth] (Lv) edge (Lsink);
        \node[draw, ellipse, scale=0.7] (Lv) at ($(Lsink)+(\b:1.75)$) {$1$};
    \path[-stealth] (Lv) edge (Lsink);
            \node[draw, ellipse, scale=0.7] (Lv) at ($(Lsink)+(\c:1.75)$) {$1$};
    \path[-stealth] (Lv) edge (Lsink);
            \node[draw, ellipse, scale=0.7] (Lv) at ($(Lsink)+(\d:1.75)$) {$1$};
    \path[-stealth] (Lv) edge (Lsink);
            \node[draw, ellipse, scale=0.7] (Lv) at ($(Lsink)+(\e:1.75)$) {$1$};
    \path[-stealth] (Lv) edge (Lsink);
            \node[draw, ellipse, scale=0.7] (Lv) at ($(Lsink)+(\f:1.75)$) {$1$};
    \path[-stealth] (Lv) edge (Lsink);
           \node[draw, ellipse, scale=0.7] (Lv) at ($(Lsink)+(\g:1.75)$) {$1$};
    \path[-stealth] (Lv) edge (Lsink);

   \node[draw, ellipse, scale=0.7] (Rv) at ($(Rsink)+(\a:1.75)$) {$1$};
        \path[-stealth] (Rv) edge (Rsink);
            \path[-stealth, dashed] (Rsink) edge[bend left=20] (Rv);

        \node[draw, ellipse, scale=0.7] (Rv) at ($(Rsink)+(\b:1.75)$) {$1$};
    \path[-stealth] (Rv) edge (Rsink);
        \path[-stealth, dashed] (Rsink) edge[bend left=20] (Rv);

            \node[draw, ellipse, scale=0.7] (Rv) at ($(Rsink)+(\c:1.75)$) {$1$};
    \path[-stealth] (Rv) edge (Rsink);
        \path[-stealth, dashed] (Rsink) edge[bend left=20] (Rv);

            \node[draw, ellipse, scale=0.7] (Rv) at ($(Rsink)+(\d:1.75)$) {$1$};
    \path[-stealth] (Rv) edge (Rsink);
        \path[-stealth, dashed] (Rsink) edge[bend left=20] (Rv);

            \node[draw, ellipse, scale=0.7] (Rv) at ($(Rsink)+(\e:1.75)$) {$1$};
    \path[-stealth] (Rv) edge (Rsink);
        \path[-stealth, dashed] (Rsink) edge[bend left=20] (Rv);

            \node[draw, ellipse, scale=0.7] (Rv) at ($(Rsink)+(\f:1.75)$) {$1$};
    \path[-stealth] (Rv) edge (Rsink);
        \path[-stealth, dashed] (Rsink) edge[bend left=20] (Rv);

           \node[draw, ellipse, scale=0.7] (Rv) at ($(Rsink)+(\g:1.75)$) {$1$};
    \path[-stealth] (Rv) edge (Rsink);
        \path[-stealth, dashed] (Rsink) edge[bend left=20] (Rv);

  };
    \node (Ld) at ($(Lsink)+(0:1)$) {$\cdots$};
    \node (Rd) at ($(Rsink)+(0:1)$) {$\cdots$};

%
%
%
\end{tikzpicture}
\caption{Left, the quiver whose moduli of representations is given by polygon space; Right, the doubled quiver used to construct hyperpolygon space.} 
\label{fig-polygon-quiver}
\end{figure} }
Carrying on with the work of \cite{steve}, through the natural action of   the group $G := P(U(r) \times (S^1)^n)$ on the star-shaped quivers in Figure \ref{fig-polygon-quiver}, there are two natural moment maps for the action of $G$ on $T^\ast \textrm{Rep}(Q)$ whose scalar multiples can be written as follows:
\begin{align}
 {Real~ moment~ map~}~  \mu_r(x,y) &= ((xx^\ast - y^\ast y)_0, |{x_1}|^{2}-|{y_1}|^{2}, \dots, |{x_n}|^{2}-|{y_n}|^{2}) \nonumber\\
 {Complex~ moment~ map~}~   \mu_c(x,y) &= ((xy)_0, y_1 x_1, \dots, y_n x_n)\nonumber
\end{align}
for $(\cdot)_0$   the trace-free part of a matrix. The \emph{polygon space} $\CP_n^r(\alpha)$ is the symplectic quotient
\begin{equation} \CP^r_n(\alpha) = \textrm{Rep}(Q) //_{\alpha} G. \nonumber\end{equation}
Similarly, the \emph{hyperpolygon space} $\CX_n^r(\alpha)$ can be defined  be the hyperk\"ahler quotient
\begin{equation} \CX^r_n(\alpha) = T^\ast \textrm{Rep}(Q) ///_{(\alpha,0)} G. \nonumber\end{equation}

Nakajima showed that polygon and
hyperpolygon spaces are examples of K\"ahler and hyperk\"ahler quiver varieties, which can be then used to understand parabolic Higgs bundles.  Although relations between quiver varieties and parabolic Higgs bundles have been long studied, many interesting  problems are still open. 
Building on a natural analogy between
  hyper-k\"ahler versions of moduli spaces of semistable
$n$-gons in complex projective space and certain parabolic Higgs bundles, Rayan and Fisher  identified in \cite{steve} the spaces $\CX^r_n(\alpha)$
with  degenerate Hitchin systems: for a hyperpolygon  $(x,y)$ in $\CX^r_n(\alpha)$, and $D = \sum_{i=1}^n p_i$   a fixed divisor on $\Sigma$ for which no two $p_i$  are the same,  
 a Higgs field on $\mathbb{P}^1$ with   vector bundle $E = \mathbb{C}^r \times \mathbb{P}^1$ is given by
\begin{equation}  
  \phi(z) = \sum_{i=1}^n \frac{\phi_i dz}{z-p_i},\nonumber
\end{equation}
for $z$   an affine coordinate on $\mathbb{P}^1$,  and where $\phi_i = x_i y_i$. 
More precisely, for parabolic $SL(n,\C)$-Higgs bundle, one has  the following general result (see \cite[Section 4]{steve} for a clear explanation of how parabolic Higgs bundles correspond to hyperpolygons):

\begin{theorem}[\cite{steve}]The space $\mathcal{X}^r_n(\alpha)$ may be identified with a moduli space of rank $r$ parabolic Higgs bundles on $\mathbb{P}^1$, such that the residue at each marked point lies in the closure of the minimal nilpotent orbit of $\mathfrak{sl}_r$, and whose underlying bundle is trivial.
\end{theorem}

Moreover,  under the identification above as shown in   \cite{steve}  there is a natural Hitchin map
whose base is an affine space of half the total dimension of hyperpolygon space.  
 
\begin{remark}The mathematical construction of the above general hyperpolygon spaces should lead to interesting applications in string theory. In particular, from \cite{tbrane}, these quivers  can be generated by an intersecting brane configuration in F-theory. The central node corresponds to the contribution from the 7-brane wrapped
over the gauge theory curve, and the satellite nodes    to the
flavour branes of the system.
\end{remark}

\subsection{Langlands duality.}
A final correspondence between subspaces of the moduli space of $G_\C$-Higgs bundle comes from Langlands duality.  One of the main reasons for considering branes is the connection to mirror symmetry and the geometric Langlands program. For $^L G_\C$   the Langlands dual group of $G_\C$, there is a correspondence between invariant polynomials for $G_\C$ and $^L G_\C$ giving an identification $\mathcal{A}_{G_\C} \simeq \mathcal{A}_{^L G_\C}$ of the Hitchin basis. The two moduli spaces $\mathcal{M}_{G_\C}$ and $\mathcal{M}_{^L G_\C}$ are then torus fibrations over a common base and their non-singular fibres are dual abelian varieties \cite{dopa}. Kapustin and Witten give a physical interpretation of this in terms of S-duality, using it as the basis for their approach to the geometric Langlands program \cite{Kap}. 

In this approach a crucial role is played by the various types of branes and their transformation under mirror symmetry.  Langlands duality exchanges branes' types according to 
\begin{align}(B,B,B) \leftrightarrow (B,A,A),\nonumber\\(A,B,A) \leftrightarrow (A,B,A),\nonumber\\ (A,A,B) \leftrightarrow (A,A,B),\nonumber\end{align}
   but  the exact correspondence is not yet known. Moreover, by studying branes of Higgs bundles through the Hitchin fibration, one may consider a Fourier-Mukai type transform to obtain the duality in terms of subspaces of the dual fibres, and there are a few things one can say about the branes constructed in Chapter \ref{chap-branes} and Langlands duality:

  \begin{itemize}
 
\item\textbf{3-manifold branes.} Recall that there is a natural   $(A,B,A)$-brane obtained in \cite{real}  as   fixed points of $i_2(\bar \partial_A, \Phi)=(f^*(  \partial_A),f^*( \Phi^*  ))$
where   $f : \Sigma \to \Sigma$ is a  real structure on $\Sigma$, and $\rho$ is the compact real form of $G_\C$. By means of the compact form $^L\rho$ of the Langlands dual group, these spaces have a  corresponding involution  on $\mathcal{M}_{^L G_\C}$ given by
\[^L{i}_2(\bar \partial_A, \Phi)=(f^*(  \partial_A),f^*( \Phi^*  ))= (f^*(^L\rho(\bar \partial_A)), -f^*( ^L\rho(\Phi) )),\]
whose fixed point set    in $\mathcal{M}_{^L G_\C}$ is conjectured in \cite{real} to be  the dual brane (for details see also \cite[Conjecture 1]{obw}).  
\item\textbf{Pseudo-real Higgs bundles.} 
  As seen in \cite{real}, pseudo-real Higgs bundles  in $ \mathcal{M}_{G_\C}$ form an $(A,A,B)$-brane and can be seen as the fixed point set of the involution
$i_3(\bar \partial_A, \Phi)=(f^* \sigma(\bar \partial_A),f^*\sigma( \Phi)),$
 for $\sigma$ a real form of the complex Lie group $G_\C$.
In this case,  it is conjectured in \cite{real} that the dual brane is the fixed point set of a corresponding involution defined as 
\[^L{i}_3(\bar \partial_A, \Phi)=(f^* \hat{\sigma}(\bar \partial_A),f^*\hat{\sigma}( \Phi)),\]
for a real structure $\hat{\sigma}$ on $^L {G_\C}$ which is conjectured to be the Nadler dual of $\sigma$ (see below for details).
 \item\textbf{Hyperpolygon spaces.} The branes arising  in terms of quivers may carry group structures, and Langlands duality can be considered for them. Some general results about duality between branes arising through quivers representations appear in \cite{vicky}, and geometric properties of hyperpolygon branes in the spirit of \cite{real} appear in \cite{poli}. 
\end{itemize}

One of the most interesting cases of Langlands duality for branes of Higgs bundles appears when considering   the $(B,A,A)$-brane of real Higgs bundles $\MG \subset \mathcal{M}_{G_\C}$, corresponding to real surface group  representations of the compact Riemann surface $\Sigma$. Recall that these branes can be seen through the Cartan involution $\theta$ of a real form $G$ of $G_c$, as the fixed point set of the involution 
\begin{equation}
i_1(\bar \partial_A, \Phi)=(\theta(\bar \partial_A),-\theta( \Phi)).
\end{equation}

 The dual $\MG ^L \subset ^L\mathcal{M}_{^L G_\C}$ must be of type $(B,B,B)$, a submanifold which is complex with respect to $I,J,K$, supporting a hyperholomorphic sheaf. From \cite{real}, the proposed dual brane should be taken as follows (for details see also \cite[Conjecture 2]{obw}):

  \begin{conjecture} \label{con}
   The support of the dual brane to the fixed point set of $i_1$ is the moduli space $\mathcal{M}_{\check{H}}\subset \mathcal{M}_{^LG_{c}}$ of $\check{H}$-Higgs bundles for $\check{H}$ the group associated to the Lie algebra $\check{\mathfrak{h}}$ appearing in \cite[Table 1]{LPS_Nadler}. 
 \end{conjecture} 
 
Whilst the above conjecture has not been proven yet, a lot of work as appeared in different directions with progress towards it. In what follows we shall mention some of this work
%
 
 \begin{itemize} \item \textbf{Branes with non-abelian data.} As mentioned before, the $(B,A,A)$-branes of $SL(m,\mathbb{H})$, $SO(2n,\mathbb{H})$ and $Sp(2m,2m)$-Higgs bundles were shown to have non-abelian intersection with the Hitchin fibration \cite{nonabelian}. The duality for these branes was studied in Branco's DPhil thesis recently, where the geometry of the whole singular fibres containing these branes was described before obtaining the dual spaces \cite{lucas}.
  \item \textbf{Quasi split real forms.} The intersection of the $(B,A,A)$-brane of $U(m,m)$-Higgs bundles with the regular fibres of the Hitchin fibration was studied in \cite{umm}, and through this description a candidate for the dual brane, agreeing with Conjecture \ref{con}, is presented in \cite{classes}. Moreover, for $m=1$, further support of the conjecture using cohomology classes and topological invariants  recently appeared in \cite{tamasnew}.
 \item\textbf{Split real forms.}  In the case of a split real form $G$ of $G_\C$, it is shown  in \cite[Theorem 4.12]{thesis} that the $(B,A,A)$-brane obtained through the involution $i_1$ lies over the  Hitchin base as 2-torsion points, and hence a 0-dimensional space in each fibre. The explicit description of the branes (as well as the monodromy action for most of the classical split real groups \cite{thesis,mono,mono1,mono2}) should present a good starting point to study mirror symmetry, noting the similarities with the case studied by Hitchin in \cite{classes}.

 \end{itemize}

\subsection{Problem set IV.}

This problem set is meant to make the  reader be acquainted with the space of polygons and hyperpolygons, and their interpretation in terms of Higgs bundles. Moreover, the reader will find the problems helpful to understand the different correspondences between Lie groups and their induced correspondences between Hitchin systems.  

\begin{problem} \label{prob1} \textbf{Isogenies.} 
 To understand the  tensor product of an $\SL(4,\R)$-Higgs bundles with itself, start with an   $\SL(4,\R)$-Higgs bundle $(E,\varphi)$ where $E$ has rank 4 and trivial determinant, and $\Phi$ is trace free and symmetric with respect to a local frame $\{e_1,\dots,e_4\}$. Then the Higgs field on $E\otimes E$ is
\[\Phi=\varphi\otimes I+ I\otimes\varphi. \]
\begin{itemize}

\item[i.] Recall that for any vector bundle $\mathbf{V}$ of rank $n$, and for all $k<n$, there is an isomorphism 
\begin{equation}\Lambda^k(\mathbf{V}^*)\otimes\Lambda^n(\mathbf{V})\longrightarrow \Lambda^{n-k}(\mathbf{V})\ .\end{equation}
Which structure arises when considering $n=4$, $k=2$, the determinant  $\Lambda^4 \mathbf{V}\simeq \mathcal{O}$, and  such that $\mathbf{V}$ has an orthogonal structure given by an isomorphism $q:\mathbf{V}\simeq \mathbf{V}^*$.

%
%
%
%
%
%
%
%
%

\item[ii.] With respect to the local frame for 
\begin{equation}\label{Esum}
E\otimes E=\Lambda^2E\oplus Sym^2(E),
\end{equation}
how does the restriction of $\Phi$ to $\Lambda^2E$ look like? does it have the structure of a $G$-Higgs field for some real group $G$?
%

\item[iii.] Compute how the whole Higgs field $\Phi$ looks like with respect to the frame in Problem \ref{prob1}\textbf{.ii.}, this is, with respect to
\[\{e_i\wedge e_j\ | i<j\}\cup\{e_i\otimes e_i\ | i=1,2,3,4\}\cup\{e_i\otimes e_j+e_j\otimes e_i\ | i<j\}.\]
\end{itemize}
%
%
%
%

\end{problem}
\begin{problem}\label{prob2}\textbf{Isogenies.} If the homomorphism $\Psi$ from \eqref{iso2} restricts to a map between real forms which respects the Cartan decompositions of the Lie algebras, then it induces a map between real Higgs bundles. Considering the morphisms between complex groups, calculate the following:
\begin{itemize}
\item[i.] Show what the isogeny between $SL(4,\C)$ and $SO(6,\C)$ implies when considered over real symplectic Higgs bundles, this over pairs $(E,\varphi)$ where  the bundle is $V\oplus V^*$ and the Higgs field has   symmetric entries \[\varphi=\left(\begin{array}{cc}0&\beta\\\gamma&0\end{array}\right).\]
%
\item[ii.] What are the relations between the eigenvalues of an $SL(4,\R)$-Higgs bundle and those of an $SO(3,3)$-Higgs bundle obtained thought the isogeny between complex Lie groups?
\end{itemize}

\end{problem}
\begin{problem} \textbf{Polygons.} We say that   $n$ ordered points $(p_1, \dots , p_n)$ on $\mathbb{P}^1$ are ${r}$-semistable if, for all $p_0 \in \mathbb{P}^1$ one has that 
\[\sum_{i\ |\ p_i = p_0} r_i \leq \sum_{i\ |\ p_i \neq p_0} r_i\]
\begin{itemize}
\item[i.] Show that the moduli space of ${r}$-semistable $n$ ordered points on $\mathbb{P}^1$  modulo the automorphisms of the projective line moduli space can be constructed via GIT as a quotient of the $\SL_2$-action on $(\mathbb{P}^1)^n$ with respect to an ample linearisation  associated to the weights (e.g. see \cite[Section 5.2]{vicky})

\item[ii.]  Describe the space of representations corresponding to polygon space for a fixed star shaped quiver and values of $d,n$. 

\item[iii.] Show that every automorphism of the star shaped quiver $Q_n$ must fix $v_0$, so 
$textrm{Aut}(Q_n) = \textrm{Aut}^+(Q_n)\cong S_n.$
\item[iv.] Construct two subgroups of $\textrm{Aut}(Q_n)$ and describe their actions on $Q_n$, as well as the quotient quivers under those group actions. 
\end{itemize}

\end{problem}
\begin{problem} \label{prob44}\textbf{Hyperpolygons.} Consider the quiver $Q$  given by
$\xymatrix@1{  1 \bullet \ar@/^/[r]^{x} & \bullet 2 \ar@{-->}@/^/[l]^{y}   } $
\begin{itemize}
\item[i.] For  $\sigma$ the covariant involution which sending the vertex $1$ to the vertex $2$, and sending the arrow $x$ to the arrow $y$, when is $d=(d_1,d_2)$ the dimension vector  $\sigma$-compatible in the sense of \cite{vicky}? 
\item[ii.] Encode the information of a $G_\C$-Higgs bundle $(E,\Phi)$ on $\mathbb{P}^1$ within the above quiver, for some particular complex or real Lie group $G_\C$.

\item[iii.] In terms of hyperpolygons, describe the moduli space of parabolic Higgs bundles associated to  Example \ref{exquiv} and its hyperk\"ahler structure?
\end{itemize}

\end{problem}
\begin{problem} \textbf{Hyperpolygons.} Having seen how polygons and hyperpolygons can be understood in terms of Higgs bundles, consider the framed Jordan quiver
\begin{equation}\mathcal{J}:=~
\xymatrix{1 \bullet \ar@{-->}@(dl,dr)_{b} \ar@(ul,ur)^{a} \ar@/_/[r]_{x} & \bullet l(1) \ar@{-->}@/_/[l]_{y} }\end{equation}

\begin{itemize}
\item[i.] Give an expression of the real and complex moment maps (up to scale) associated to such quivers following \cite{steve}. 
\item[ii.] Let $\sigma:\mathcal{J}\rightarrow \mathcal{J}$ be the involution of the Jordan quiver which exchanges the vertices, and exchanges the arrows. Through the induced action on the moduli space, give a description of the fixed point set. 
\item[iii.] As mentioned before, the space of hyperpolygons studied in the previous sections inherits a hyperk\"ahler structure and thus one may consider three fixed complex structures and their associated symplectic forms, and how these relate to a given subspace. Following \cite{vicky} and \cite{poli}, decide whether the fixed point set in Problem \ref{prob44}\textbf{.ii.} is a brane, and if so, give its type. 
\end{itemize}

\end{problem}
 
\bigskip

\noindent \textbf{Acknowledgments}. The author would  like to thank the Institute of Advance Studies and the organisers of the  \textit{   2019 Graduate Summer School}, as well as the director of PCMI Rafe Mazzeo,  for providing an ideal environment for research and collaboration. We are also thankful to Steve Rayan, Charles Alley, Melissa Zhang, Mengxue Yang, Szilard Szabo, and Richard Derryberry,  for useful comments on a draft of the notes; to Siqi He and Richard Derryberry for their help running the course, and to Ian Morrison for his help with the notes.   The author also acknowledges support from the Humboldt Foundation as well as from the U.S. National Science Foundation grants   DMS-1509693 and CAREER DMS 1749013.

%
%
%
%
             \bibliography{yourbibfile}

%
%
%
%
\bibspread

\end{document}